\documentclass[a4paper,UKenglish,numberwithinsect,thm-restate]{lipics-v2021}
\pdfoutput=1 %

\usepackage[ruled,linesnumbered,noline,noend]{algorithm2e}
\usepackage{mathtools} %
\usepackage{stmaryrd} %
\usepackage{dsfont} %
\usepackage{tikz}

\usepackage{amsmath, amssymb}
\usepackage{tikz}
\usetikzlibrary{tikzmark}

\usetikzlibrary{automata}
\usetikzlibrary{positioning}
\usepackage[ruled,linesnumbered,noline,noend]{algorithm2e}

\usepackage[textsize=small]{todonotes}

\usepackage{hyperref}
\hypersetup{
    pdfencoding=auto,
    colorlinks,
    linkcolor={blue!70!black},
    citecolor={blue!70!black},
    urlcolor={blue!70!black}
}

\DeclarePairedDelimiter\braces{\{}{\}}
\DeclarePairedDelimiter\minor{\lfloor}{\rfloor}

\DeclarePairedDelimiter\denote{\llbracket}{\rrbracket}

\newcommand{\ck}{\mathit{c}^K}
\newcommand{\tss}{\mathit{T}\Sigma^{*}}
\newcommand{\Aa}{\mathcal{A}}
\newcommand{\Bb}{\mathcal{B}}
\newcommand{\xra}{\xrightarrow}
\newcommand{\Ll}{\mathcal{L}}
\newcommand{\ma}{{\sffamily\bfseries\upshape (a)\ }}
\newcommand{\mb}{{\sffamily\bfseries\upshape (b)\ }}

\newcommand{\lm}{\lambda}
\newcommand{\mone}{{\sffamily\bfseries\upshape \color{gray}{1.} \ }}

\title{A Myhill-Nerode style Characterization for Timed Automata With Integer Resets}
\author{Kyveli Doveri}{University of Warsaw, Poland}{k.doveri@mimuw.edu.pl}{https://orcid.org/0000-0001-9403-2860}{Supported by the ERC grant INFSYS, agreement no. 950398. Results partly obtained when previously affiliated with the IMDEA Software Institute.}
\author{Pierre Ganty}{IMDEA Software Institute, Pozuelo de Alarc\'{o}n, Spain}{pierre.ganty@imdea.org}{https://orcid.org/0000-0002-3625-6003}{This publication is part of the grant PID2022-138072OB-I00, funded by \linebreak MCIN/AEI/10.13039/501100011033/FEDER, UE.}
\author{B. Srivathsan}{Chennai Mathematical Institute, India\\ CNRS IRL 2000, ReLaX, Chennai, India}{sri@cmi.ac.in}{https://orcid.org/0000-0003-2666-0691}{}

\authorrunning{K. Doveri, P. Ganty, B. Srivathsan} %

\Copyright{Kyveli Doveri, Pierre Ganty and B. Srivathsan} 


\category{} %

\relatedversion{} %

\nolinenumbers %
\hideLIPIcs %

\EventEditors{Siddharth Barman and S{\l}awomir Lasota}
\EventNoEds{2}
\EventLongTitle{44th IARCS Annual Conference on Foundations of Software Technology and Theoretical Computer Science (FSTTCS 2024)}
\EventShortTitle{FSTTCS 2024}
\EventAcronym{FSTTCS}
\EventYear{2024}
\EventDate{December 16--18, 2024}
\EventLocation{Gandhinagar, Gujarat, India}
\EventLogo{}
\SeriesVolume{323}
\ArticleNo{22}

\newcommand{\mn}{\sim_{\scriptscriptstyle L}}
\newcommand{\eqv}{\mathord{\approx^{L, K}}}
\begin{document}
\maketitle
\begin{abstract}
  The well-known Nerode equivalence for finite words plays a
  fundamental role in our understanding of the class of regular
  languages.  The equivalence leads to the Myhill-Nerode theorem and a
  canonical automaton, which in turn, is the basis of several automata
  learning algorithms.  A Nerode-like equivalence has been studied for
  various classes of timed languages.

  In this work, we focus on timed automata with integer resets.  This
  class is known to have good automata-theoretic properties and is
  also useful for practical modeling.  Our main contribution is a
  Nerode-style equivalence for this class that depends on a constant
  \(K\).  We show that the equivalence leads to a Myhill-Nerode
  theorem and a canonical one-clock integer-reset timed automaton with
  maximum constant \(K\).  Based on the canonical form, we develop an
  Angluin-style active learning algorithm whose query complexity is
  polynomial in the size of the canonical form.
\end{abstract}
\newpage
\setcounter{page}{1}

\section{Introduction}
A cornerstone in our understanding of regular languages is the
Myhill-Nerode theorem.  This theorem characterizes regular languages
in terms of the Nerode equivalence \({\mn}\): for a word \(w\) we
write \(w^{-1} L = \{ z \mid wz \in L \}\) for the \emph{residual
  language} of \(w\) w.r.t. \(L\); and for two words \(u,v\) we say
\(u \mn v\) if \(u^{-1} L = v^{-1} L\).

\begin{theorem}[Myhill-Nerode theorem]\label{thm:myhill-nerode}
  Let \(L\) be a language of finite words.
  \begin{itemize}

  \item \(L\) is regular if{}f the Nerode equivalence has a finite
    index.

  \item The Nerode equivalence is coarser than any other monotonic
    \(L\)-preserving equivalence.
  \end{itemize}
\end{theorem}

An equivalence is said to be \emph{monotonic} if \(u \approx v\)
implies \(u a \approx v a\) for all letters \(a\) and is
\(L\)-preserving if each equivalence class is either contained in
\(L\) or disjoint from \(L\).
\footnote{The exact term would be ``right monotonic'' because it only
  considers concatenation to the right of the word.  Throughout the
  paper we simply write monotonic to keep it short but we mean right
  monotonic.  } An equivalence over words being monotonic makes it
possible to construct an automaton with states being the equivalence
classes.  The Nerode equivalence being the coarsest makes the
associated automaton the minimal (and a canonical) deterministic
automaton for the regular language.  Our goal in this work is to
obtain a similar characterization for certain subclasses of timed
languages.

Timed languages and timed automata were introduced by Alur and
Dill~\cite{AlurDTimedAutomata94} as a model for systems with real-time
constraints between actions.  Ever since its inception, the model has
been extensively studied for its theoretical aspects and practical
applications.  In this setting, words are decorated with a delay
between consecutive letters.  A \emph{timed word} is a finite sequence
\((t_1 \cdot a_1) (t_2 \cdot a_2) \cdots (t_n \cdot a_n)\) where each
\(t_i \in \mathbb{R}_{\ge 0}\) and each \(a_i\) is a letter taken from
a finite set \(\Sigma\) called an \emph{alphabet}.  A timed word
associates a time delay between letters: \(a_1\) was seen after a
delay of \(t_1\) from the start, the next letter \(a_2\) appears
\(t_2\) time units after \(a_1\), and so on.  Naturally, a timed
language is a set of timed words.  A timed automaton is an automaton
model that recognizes timed languages.  Figures~\ref{fig:ta_I} and
\ref{fig:equiv-words-diff-states} present some examples (formal
definitions appear later).  In essence, a timed automaton makes use of
\emph{clocks} to constrain time between the occurrence of transitions.
In Figure~\ref{fig:ta_I}, the variable \(x\) denotes a clock.  The
transition labels are given by triples comprising a letter
(e.g. \(a\)), a clock constraint (e.g. \(x=1\)), and a multiplicative
factor (\(0\) or \(1\)) for the clock update.  Intuitively, the
semantics of the transition from \(q_I\) to \(q\) is as follows: the
automaton reads the letter \(a\) when the value of the clock held in
\(x\) is exactly \(1\) and updates the clock value to \(1 \times x\).
If the third element of the transition label is \(0\), then the
transition updates the value of \(x\) to \(0 \times x = 0\).  We refer
to the second element of the transition label as the transition
\emph{guard} and the third element as the \emph{reset}.  It is worth
mentioning that the transition guards feature constants given by
integer values, meaning that a guard like \(x=0.33\) is not allowed.
Next, we argue how challenging it is to define a Nerode-style
equivalence for timed languages.

\subparagraph*{Challenge 1.} \emph{The Nerode equivalence lifted as it
  is has infinitely many classes.} %
For example, the timed automaton of Figure~\ref{fig:ta_I} accepts a
timed word \( (t_1\cdot a)\ldots(t_n \cdot a)\) as long as
\(t_1+\dots +t_n=1\).  The timed language \(L\) of that automaton has
infinitely many quotients.  Indeed let \(0 < t_1 < 1\), we have that
\(( t_1\cdot a)^{-1} L=\{ (t_2\cdot a)\ldots (t_n\cdot a) \mid
t_2+\dots +t_n=1-t_1\}\).  Observe that different values for \(t_1\)
yield different quotients, hence \(L\) has uncountably many quotients.

\begin{figure}
  \centering
  \begin{tikzpicture}[shorten >=1pt,node distance=3cm,on grid,auto]
    \node[state,initial right,initial text={}] (epsilon) {\(q_{I}\)};
    \node[state,accepting] (q) [right=of epsilon] {\(q\)}; \path[->]
    (epsilon) edge [bend left=12] node {\(a,x=1,1\)} (q) (epsilon)
    edge [loop left] node [left] {\(a,x<1,1\)} () (q) edge [loop
    right] node [right] {\(a,x=1,1\)} ();
  \end{tikzpicture}
  \caption{Automaton accepting
    \(L=\{ (t_1\cdot a)\ldots (t_n\cdot a) \mid t_1+\dots +t_n=1\}\)
    with alphabet \(\Sigma=\{a\}\).}
  \label{fig:ta_I}
\end{figure}
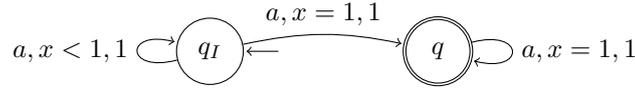

\subparagraph*{Challenge 2.} \emph{Two words with the same residual
  languages may never go to the same control state in any timed
  automaton.}  Figure~\ref{fig:equiv-words-diff-states} gives an
example of a timed language that exhibits this challenge.
Consider the words \(u = (0.5 \cdot a)\) and \(v = (1.5 \cdot a)\).
The residual of both these words is the singleton language
\(\{ (0.5 \cdot b) \}\).  Suppose both \(u\) and \(v\) go to the same
control state \(q\) in the timed automaton.  After reading \(u\)
(resp.  \(v\)), clocks which are possibly reset will be \(0\), whereas
the others will be \(0.5\) (resp. \(1.5\)).  Suppose \(v\) is accepted
via a transition sequence \(q_I \xrightarrow{} q \xrightarrow{} q_F\).
Since guards contain only integer constants, the guard on
\(q \xrightarrow{} q_F\) should necessarily be of the form \(x = 2\)
for some clock \(x\) which reaches \(q\) with value \(1.5\).  The same
transition can then be taken from \(u\) to give \(u (2 \cdot b)\) or
\(u (1.5 \cdot b)\) depending on the value of \(x\) after reading
\(u\).  A contradiction.  This example shows there is no hope to
identify states of a timed automaton through quotients of a
Nerode-type equivalence.  The equivalence that we are aiming for needs
to be stronger, and further divide words based on some past history.

\subparagraph*{Challenge 3.} \emph{The Nerode-style equivalence should
  be amenable to a timed automaton construction.}  In the case of
untimed word languages, monotonicity of the Nerode equivalence
immediately led to an automaton construction.  We need to find the
right notion of monotonicity for the class of automata that we want to
build from the equivalence.

\begin{figure}
  \centering
  \begin{tikzpicture}[shorten >=1pt,node distance=3cm,on
    grid,auto,initial text={}]
    \node[state, initial] (0) at (0,0) {\(q_0\)}; \node[state] (1) at
    (3,.5) {\(q_1\)}; \node[state] (2) at (3,-.5) {\(q_2\)};
    \node[state, accepting] (3) at (6,0) {\(q_3\)};
    \begin{scope}[->, auto]
      \draw (0) to node [sloped] {\(a, 0 < x < 1,1\)} (1); \draw (0)
      to node [below,sloped] {\(a, 1 < x < 2,1\)} (2); \draw (1) to
      node [sloped] {\(b, x = 1,0\)} (3); \draw (2) to node
      [below,sloped] {\(b, x = 2,0\)} (3);
    \end{scope}
  \end{tikzpicture}
  \caption{Automaton accepting
    \(L_2 = \{(t_1 \cdot a) (t_2 \cdot b) \mid \text{either } 0 < t_1
    < 1 \text{ and } t_1 + t_2 = 1, \text{ or } 1 < t_1 < 2 \text{ and
    } t_1 + t_2 = 2 \}\) with alphabet \(\Sigma=\{a,b\}\)}
  \label{fig:equiv-words-diff-states}
\end{figure}
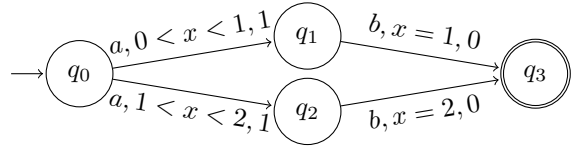

A machine independent characterization for deterministic timed
languages has been studied by Boja\'{n}czyk and
Lasota~\cite{bojanczykMachineIndependentCharacterizationTimed2012}.
They circumvented the above challenges by considering a new automaton
model \emph{timed register automata} that generalizes timed automata.
This automaton model makes use of registers to store useful
information, for instance for the language in
Figure~\ref{fig:equiv-words-diff-states}, a register stores the value
\(0.5\) after reading \((0.5 \cdot a)\) and \((1.5 \cdot a)\).  This
feature helps in resolving Challenge 2.  For the question of
finiteness mentioned in Challenge 1, timed register automata are
further viewed as a restriction of a more general model of automata
that uses the abstract concept of \emph{Frankel-Mostowski} sets in its
definition.  Finiteness is relaxed to a notion of
\emph{orbit-finiteness}.

The work of An et al.~\cite{anLearningNondeterministicRealTime2021}
takes another approach to these challenges by considering a subclass
of timed languages which are called \emph{real-time languages}.  These
are languages that can be recognized using timed automata with a
single clock that is reset in every transition.  Therefore, after
reading a letter, the value of the clock is always \(0\).  This helps
in solving the challenges, resulting in a canonical form for real-time
languages.

\subparagraph*{Our work.}  As we have seen, to get a characterization
which also lends to an automaton construction, either the automaton
model has been modified or the characterization is applied to a class
of languages where the role of the clock is restricted to consecutive
letters.  Our goal is to continue working with the same model as timed
automata and apply a characterization to a different subclass.

In this work, we look at languages recognized by timed automata with
integer resets (IRTA).  These are automata where clock resets are
restricted to transitions that contain a guard of the form \(x = c\)
for some clock \(x\) and some integer
\(c\)~\cite{sumanTimedAutomataInteger2008}.  The class of languages
recognized by IRTA is incomparable with real-time languages.
Moreover, it is known that IRTA can be reduced to
\(1\)-clock-deterministic IRTA~\cite{manasaIntegerResetTimed2010}, or
\(1\)-IRDTA for short.  The proof of this result effectively computes,
given an IRTA, a timed language equivalent \(1\)-IRDTA.  Here is our
main result which gives a Myhill-Nerode style characterization for
IRTA languages.
\begin{theorem}
  Let \(L\) be a timed language.
  \begin{itemize}
  \item \(L\) is accepted by a timed automaton with integer resets
    if{}f there exists a constant \(K\) such that \(\eqv\) is
    \(K\)-monotonic and has a finite index.
  \item The \(\eqv\) equivalence is coarser than any \(K\)-monotonic
    \(L\)-preserving equivalence.
  \end{itemize}
\end{theorem}

Intuitively, one should think of \(K\) as the largest integer that
needs to appear in the guards of an accepting automaton.  The goal of
the paper is to identify the notion of \(K\)-monotonicity and the
equivalence \(\eqv\) that exhibit the above theorem.  The
characterization also leads to a canonical form for IRTA.  In
practice, the integer reset assumption allows for modeling multiple
situations \cite{sumanTimedAutomataInteger2008}.

To the best of our knowledge, there is no learning algorithm that can
compute an IRTA for systems that are known to satisfy the integer
reset assumption.  We fill this gap and show how Angluin's style
learning \cite{angluinLearningRegularSets1987} can be adapted to learn
\(1\)-IRDTA.

\subparagraph*{Related work.}  Getting a canonical form for timed
languages has been studied in several works:
\cite{bojanczykMachineIndependentCharacterizationTimed2012} and
\cite{malerRecognizableTimedLanguages2004a} focus on a machine
independent characterization for deterministic timed languages,
whereas the works
\cite{grinchteinLearningEventrecordingAutomata2010,anLearningNondeterministicRealTime2021,wagaActiveLearningDeterministic2023}
extend the study of the canonical forms to an active learning
algorithm.  Languages accepted by
event-recording automata are a class of languages where the value of
the clocks is determined by the input word.  This helps in coming up
with a canonical
form~\cite{grinchteinLearningEventrecordingAutomata2010}.  In
\cite{wagaActiveLearningDeterministic2023}, the author presents a
Myhill-Nerode style characterization for deterministic timed languages
by making use symbolic words rather than timed words directly.  The
author shows that the equivalence has a finite index if{}f the
language is recognizable (under the notion of recognizability using
right-morphisms proposed by Maler and Pnueli
\cite{malerRecognizableTimedLanguages2004a}).  Further, Maler and
Pnueli have given an algorithm to convert recognizable timed languages
to deterministic timed automata, which resets a fresh clock in every
transition and makes use of clock-copy updates \(x: = y\) in the
transitions.  It is known that automata with such updates can be
translated to classical timed
automata~\cite{SpringintveldVMinimizableTA1996,
  BouyerDFPUpdatableTA2004}.

Learning timed automata is a topic of active research.  The
foundations of timed automata learning were laid in the pioneering
work of Grinchtein et al.
\cite{grinchteinLearningEventrecordingAutomata2010} by providing a
canonical form for event-recording automata (ERAs).  These are
automata having a clock for each letter in the alphabet, and a clock
$x_a$ records the time since the last occurrence of $a$.  The
canonical form essentially considers a separate state for
each region.  Since there are as many clocks as the number of letters,
there are at least $|\Sigma|!  $ number of regions.  This makes the
learning algorithm prohibitively expensive to implement.  In contrast,
as we will see, we are able to convert IRTAs into a subclass of
single-clock IRTAs.  If $K$ is the maximum constant, there are only
$2K + 2$ many regions.  Later works on learning ERAs have considered
identifying other forms of automata that merge the states of the
canonical
form~\cite{DBLP:conf/atva/LinADSL11,DBLP:conf/concur/GrinchteinJP06,Sayan}.
Other models for learning timed systems consider one-clock timed
automata~\cite{DBLP:conf/atva/XuAZ22, DBLP:conf/tacas/AnCZZZ20} and
Mealy machines with timers~\cite{DBLP:journals/corr/abs-2403-02019,
  DBLP:journals/iandc/VaandragerEB23}.  Approaches other than active
learning for timed automata include passive learning of discrete timed
automata~\cite{DBLP:phd/basesearch/Verwer10} and learning timed
automata using genetic
programming~\cite{DBLP:conf/formats/TapplerALL19}.

We have considered a subclass of deterministic timed languages.
Therefore, our class does fall under the purview of the
\cite{wagaActiveLearningDeterministic2023,malerRecognizableTimedLanguages2004a}
work --- however, the fundamental difference is that we continue to
work with timed words and not symbolic timed words.  This gives an
alternate perspective and a direct and simpler \(1\)-clock IRTA
construction.  The simplicity and directness also apply when it comes
to learning \(1\)-clock IRTA.

\subparagraph*{Outline of the paper.}  In
Section~\ref{sec:integer-resets}, we define the class of one-clock
languages with integer resets, and their acceptors thereof: the
1-clock Timed Automata (or \(1\)-TA) with clock constraints given by
region equivalence classes and transitions that always reset on
integer clock values.  Section~\ref{sec:from-equiv-autom} puts forward
a notion of $K$-monotonicity and characterizates $K$-monotonic 
equivalences as a certain type of integer reset automata.
Subsequently, Section~\ref{sec:nerode-style-equiv} presents the
Nerode-style equivalence, the Myhill-Nerode theorem and some examples
applying the theorem.  Finally, in Section~\ref{sec:canonical-form},
we give an algorithm to compute and learn the canonical form.

\section{Background}
\label{sec:background}
\subparagraph*{Words and languages.} %
An \emph{alphabet} is a finite set of \emph{letters} which we
typically denote by \(\Sigma\).  An \emph{untimed word} is a finite
sequence \(a_{1}\cdots a_{n}\) of letters \(a_{i}\in \Sigma\).  We
denote by \(\Sigma^{*}\) the set of untimed words over \(\Sigma\).  An
\emph{untimed language} is a subset of \(\Sigma^{*}\).  A
\emph{timestamp} is a finite sequence of non-negative real numbers.
We denote the latter set by \(\mathbb{R}_{{}\geq 0}\) and the set of
all timestamps by \(\mathbb{T}\).  A \emph{timed word} is a finite
sequence \((t_1\cdot a_{1})\cdots (t_n\cdot a_{n})\) where
\(a_{1}\cdots a_{n}\in \Sigma^{*}\) and
\(t_1\cdots t_n\in \mathbb{T}\).  We denote the set of timed words by
\(\mathbb{T}\Sigma^{*}\).  Given a timed word
\(u = (t_1\cdot a_1)\ldots(t_n\cdot a_n)\in\mathbb{T}\Sigma^{*}\) we
denote by \(\sigma(u)\) the sum \(t_1+\cdots+t_n\) of its timestamps.
A \emph{timed language} is a set of timed words.  As usual, we denote
the empty (un)timed word by \(\epsilon\).  The \emph{residual
  language} of an (un)timed language \(L\) with regard to a (un)timed
word \(u\) is defined as \(u^{-1}L=\{w\mid uw \in L\}\).  Therefore it
is easy to see that \(\epsilon^{-1} L = L\) for every (un)timed
language \(L\).

Timed automata~\cite{AlurDTimedAutomata94} are recognizers of timed
languages.  Since we focus on subclasses of timed automata with a
single clock, we do not present the definition of general timed
automata.  Instead, we give a modified presentation of one-clock timed
automata that will be convenient for our work.

\subparagraph*{One-clock timed automata.} %
A \emph{One-clock Timed Automaton} (\(1\)-TA) over \(\Sigma\) is a
tuple \(\mathcal{A} = (Q, q_{I},T,F)\) where \(Q\) is a finite set of
states, \(q_{I}\in Q\) is the initial state, \(F\subseteq Q\) is the
set of final states and
\(T\subseteq Q\times Q\times \Sigma \times \Phi \times \{0,1\}\) is a
finite set of transitions where \(\Phi\) is the set of clock
constraints given by
\[ \phi ::= x < m \quad|\quad m < x \quad|\quad x=m \quad|\quad \phi
  \wedge \phi \enspace,\text{ where }m\in \mathbb{N}.
\]

For a clock constraint \(\phi\), we write \(\denote{\phi}\) for the
set of non-negative real values for \(x\) that satisfies the
constraint.  Notice that we have disallowed guards of the form
\(x \le m\) which appear in standard timed automata literature, since
its effect can be captured using two transitions, one with \(x = m\)
and another with \(x < m\).  A transition is a tuple
\((q, q', a, \phi, r)\) where \(\phi\) is a clock constraint called
the \emph{guard} of the transition and \(r \in \{0, 1\}\) denotes
whether the single clock \(x\) is \emph{reset} in the transition.

We say that a \(1\)-TA with transitions \(T\) is \emph{deterministic}
whenever for every pair \(\theta=(q,q',a,\phi,r)\) and
\( \theta_1=(q_1,q'_1,a_1,\phi_1,r_1)\) of transitions in \(T\) such
that \(\theta\neq\theta_1\) we have that either \(q\neq q_1\),
\(a\neq a_1\) or \(\denote{\phi}\cap\denote{\phi_1}=\emptyset\).

A \emph{run} of \(\mathcal{A}\) on a timed word
\((t_1\cdot a_{1})\ldots (t_k\cdot a_{k})\in \mathbb{T}\Sigma^{*}\) is
a finite sequence
\[ e=(q_{0}, \nu_{0}) \xrightarrow{t_1,\theta_{1}} (q_{1},
  \nu_{1})\xrightarrow{t_2,\theta_{2}}\cdots
  \xrightarrow{t_k,\theta_{k}} (q_{k}, \nu_{k})\enspace, \] where
\(\braces{q_0,\ldots,q_k}\subseteq Q\),
\(\braces{\nu_0,\ldots,\nu_k}\subseteq \mathbb{R}_{\ge 0}\) and for
each \(i \in \{1,\ldots,k\}\) the following hold: \(\theta_i\in T\)
and \(\theta_{i}\) is of the form
\((q_{i-1},q_{i},a_{i},\phi_{i},r_{i})\),
\(\nu_{i-1}+t_i\in \denote{\phi_{i}}\), and
\(\nu_{i} = r_{i}(\nu_{i-1}+t_i)\).  Therefore if \(r_i = 0\), we have
\(\nu_i = 0\) and if \(r_i = 1\) we have \(\nu_i = \nu_{i-1} + t_i\).
A pair \( (q,\nu) \in Q\times\mathbb{R}_{\ge 0}\) like the ones
occurring in the run \(e\) is called a \emph{configuration} of
\(\mathcal{A}\) and the configuration \( (q_I,0) \) is called
\emph{initial}.  The run \(e\) is deemed \emph{accepting} if
\(q_{k}\in F\).

For \(w\in \mathbb{T}\Sigma^{*}\) we write
\((q,\nu) \rightsquigarrow^w (q',\nu')\) if there is a run of
\(\mathcal{A}\) on \(w\) from \((q,\nu)\) to \((q',\nu')\).  Observe
that if \(\mathcal{A}\) is deterministic then for every timed word
\(w\) there is at most one run on \(w\) starting from the initial
configuration.  We say that \(\mathcal{A}\) is complete if every word
admits a run.  In the rest, we will always assume, without loss of
generality, that our timed automata are complete.  Finally, given a
configuration \((q,x)\) define
\(\mathcal{L}(q,x)=\{w\in \mathbb{T}\Sigma^{*}\mid (q,x)
\rightsquigarrow^w (p,\nu), p\in F, \nu \in \mathbb{R}_{\ge 0}\}\),
hence define \(L(\mathcal{A})\) as \(\mathcal{L}(q_I,0)\).

\subparagraph*{Equivalence relation.}  A relation
\(\mathord{\sim} \subseteq S\times S\) on a set \(S\) is an
\emph{equivalence} if it is reflexive (i.e. \(x\sim x\)), transitive
(i.e. \(x\sim y \wedge y\sim z \implies x\sim z\)) and symmetric
(i.e. \(x\sim y\implies y\sim x\)).  The equivalence class of
\(s\in S\) w.r.t. \(\sim\) is the subset
\([s]_{\sim}=\{s'\in S\mid s\sim s'\}\).  A \emph{representative} of
the class \([s]_{\sim}\) is any element \(s'\in [s]_{\sim}\).  Given a
subset \(D\) of \(S\) we define
\([D]_{\sim}=\{[d]_{\sim}\mid d\in D\}\).  We say that \(\sim \) has
\emph{finite index} when \([S]_{\sim}\) is a finite set.  An important
notion in the analysis of timed automata is the \emph{region
  equivalence} which we recall next for one-clock timed automata.

\subparagraph*{Region equivalence.}  Given a constant
\(K\in\mathbb{N}\) define the equivalence
\(\mathord{\equiv^{K}}\subseteq \mathord{\mathbb{R}_{\ge 0}\times
  \mathbb{R}_{\ge 0}}\) by
\[x\equiv^{K} y \iff \bigl( \minor{x} = \minor{y} \land (\{x\}=0
  \Leftrightarrow \{y\}=0) \bigr) \lor (x > K \land y > K) \enspace
  ,\] where given \(x \in \mathbb{R}_{\ge 0}\) we denote by
\(\minor{x} \) its integral part and by \(\{x\}\) its fractional part.
\section{Languages with Integer Resets}
\label{sec:integer-resets}
We are interested in timed languages recognized by IRTAs. It is known
that IRTAs can be converted to 1-clock deterministic
IRTAs~\cite{manasaIntegerResetTimed2010}. The key idea is that in any
reachable valuation of an IRTA, all clocks have the same fractional
value. Therefore, the integral values of all clocks can be encoded inside
the control state, and the fractional values can be read from a single
clock. In the sequel, we simply define IRTAs with a
single clock, due to the equi-expressivity.

We define the class of one-clock integer-reset timed automata
(\(1\)-IRTA) where transitions reset the clock provided its value is
an integer.
Formally, we say that a \(1\)-TA \(\mathcal{A} = (Q, q_{I},T,F)\) is a
\(1\)-IRTA when for every resetting transition
\((q,q',a,\phi,0)\in T\) the clock constraint \(\phi\) is of the form
\(x=m\), or, equivalently, \(\denote{\phi}\in\mathbb{N}\).  A
deterministic \(1\)-IRTA is called a \(1\)-IRDTA.
\begin{example}
  The one-clock timed automata in Figures~\ref{fig:ta_I},
  \ref{fig:equiv-words-diff-states} and \ref{fig:ta_R_int_not_I} are
  all \(1\)-IRTAs.
\end{example}
The definition of a run of a \(1\)-IRTA on a timed word simply follows
that of \(1\)-TA.  However, for \(1\)-IRTA, we can identify positions
in input timed words where resets can potentially happen.  The next
definition makes this idea precise.
\begin{definition}
  \label{def:c^K}
  Given \(d=t_1\cdots t_n\in \mathbb{T}\) and a \(K \in \mathbb{N}\),
  we define the longest sequence of indices
  \(s_{d}=\{0=i_{0}< i_{1}< \dots <i_{p}\leq n\}\) such that for every
  \(j\in \{0, \ldots, p-1\}\) the value
  \(\sum_{i=(i_{j})+1}^{i_{(j+1)}} t_i\) is an integer between \(0\)
  and \(K\).  We refer to the set of positions of the sequence \(s_d\)
  as the \emph{integral positions} of \(d\).  Note that \(s_d\) is
  never empty since it always contains \(0\).  Next, define
  \[ \mathit{c}^{K}(d) = \sum_{i=(i_{p})+1}^{n} t_i\enspace .
  \]
  The definitions of integral positions and the function
  \(\mathit{c}^K\) apply equally to timed words by taking the
  timestamp of the timed word.
\end{definition}
Notice that the sequence \(s_{d}\) depends on the constant \(K\) (we
do not explicitly add \(K\) to the notation for simplicity, as in our
later usage, \(K\) will be clear from the context).  Note also that
when \(i_p=n\) then \(\mathit{c}^{K}(d)=0\), otherwise
\(\mathit{c}^{K}(d)\) can be any real value except an integer value
between \(0\) and \(K\), i.e.
\(\mathit{c}^{K}(d)\in \mathbb{R}_{\geq 0}\setminus\{0,\ldots,K\}\).
\begin{example}
  For \(K=1\) and \(u=(0.2\cdot a)(0.8\cdot a)(0.2\cdot a)\) we have
  \(s_{u}=\{0<2\}\) and \(\mathit{c}^{1}(u)=0.2\).  For \(K=1\) and
  \(u'=(1.2\cdot a)(0.8\cdot a)(0.2\cdot a)\) we have \(s_{u'}=\{0\}\)
  and \(\mathit{c}^{1}(u')=2.2\).  For \(K = 2\), we have
  \(s_{u'} = \{ 0 < 2 \}\) and \(\mathit{c}^{2}(u')=0.2\).
\end{example}
Consider a run of a \(1\)-IRTA on a word
\(u=(t_1\cdot a_{1})\cdots (t_n\cdot a_{n})\in \mathbb{T}\Sigma^{*}\)
and factor it according to
\(s_{u}=\{0=i_{0}< i_{1}< \dots <i_{p}\leq n\}\):
\begin{multline*}
  (q_{i_0}, \nu_{i_0}) \xrightarrow{t_{(i_0)+1},\theta_{(i_0)+1}\cdots
    t_{i_{1}},\theta_{i_1}} (q_{i_{1}},
  \nu_{i_{1}})\xrightarrow{t_{(i_1)+1},\theta_{(i_1)+1} \cdots
    t_{i_2},\theta_{i_2}} (q_{i_{2}} \nu_{i_{2}})
  \xrightarrow{}\cdots\\ \xrightarrow{} (q_{i_{p}}, \nu_{i_{p}})
  \xrightarrow{t_{(i_{p})+1},\theta_{(i_{p})+1} \cdots t_n,\theta_{n}}
  (q_{n}, \nu_{n})
\end{multline*}
At each position \(i_{j}\) with \(j\in\{0,\ldots,p\}\),
\(\nu_{i_j}\in\mathbb{N}\) and, moreover, \(\nu_{i_j}=0\) when
\(r_{i_j}=0\).

In Definition~\ref{def:c^K} we identified the integral positions at
which a \(1\)-IRTA \emph{could potentially} reset the clock.  In the
following, we recall a subclass of \(1\)-IRTAs called strict
\(1\)-IRTAs \cite{bhaveAddingDenseTimedStack2017} where every
transition with an equality guard \(\phi\) (\(\phi\) is of the form
\(x = m\) or, equivalently, \(\denote{\phi}\in\mathbb{N}\))
\emph{must} reset the clock.  This feature, along with a special
requirement on guards forces a reset on every position given by
\(s_u\) for a word \(u\).

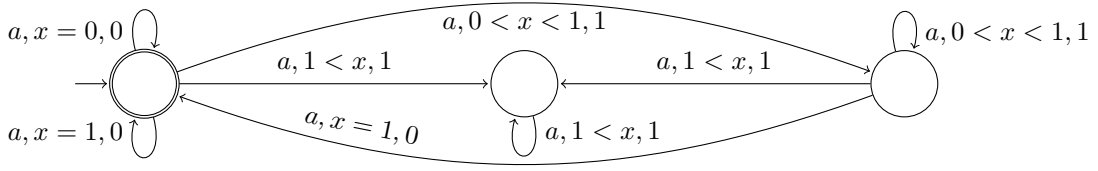
\begin{figure}
  \centering
  \begin{tikzpicture}[shorten >=1pt,node distance=6cm,on
    grid,auto,initial text={}]
    \node[state,initial, accepting] (epsilon) at (0,0) {};
    \node[state] (q) at (10,0) {}; \node[state] (2) at (5,0) {};
    \path[->] (epsilon) edge [loop above] node [left, very near start]
    {\(a,x=0,0\)} () (epsilon) edge [bend left=20, below] node
    {\(a,0<x<1,1\)} (q) (epsilon) edge [loop below] node [left, very
    near end] {\(a,x=1,0\)} () (epsilon) edge node {\(a,1<x,1\)} (2)
    (q) edge [loop above] node [right, very near end] {\(a,0<x<1,1\)}
    () (q) edge [bend left=20, above] node [sloped, near end]
    {\(a,x=1,0\)} (epsilon) (q) edge [above] node {\(a,1<x,1\)} (2)
    (2) edge [loop below] node [right, very near start] {\(a,1<x,1\)}
    ();
  \end{tikzpicture}
  \caption{A strict \(1\)-IRTA with alphabet \(\Sigma=\{a\}\)
    accepting
    \(M=\{u \in \mathbb{T}\Sigma^{*}\mid \mathit{c}^{1}(u)=0\}\).}
  \label{fig:ta_R_int_not_I}
\end{figure}

\subsection{Strict \texorpdfstring{\(1\)}{1}-IRTA}
\label{sec:strict-IRTAs}
A \(1\)-IRTA is said to be \emph{strict} if there exists
\(K\in\mathbb{N}\) such that for each of its transitions
\( (q,q',a,\phi,r) \) the following holds:
\begin{enumerate}
\item the clock constraint of the guard \(\phi\) is either \(x=m\),
  \( m < x \land x < m+1\), or \(K<x\),
\item the clock constraint of the guard \(\phi\) is an equality if{}f
  \(r=0\).
\end{enumerate}
\begin{example}
  \label{example:strict}
  The \(1\)-IRTA in Figures~\ref{fig:equiv-words-diff-states}
  and~\ref{fig:ta_R_int_not_I} are strict \(1\)-IRTAs whereas the one
  in Figure~\ref{fig:ta_I} is not strict since the transition
  \(q_I \xrightarrow{a, x = 1, 1} q\) does not reset the clock.  To
  make it strict while accepting the same language, we replace the
  transitions of guard $x=1$ by resetting transitions of guard $x=0$
  and, split the transition \(q_I \xrightarrow{a, x<1, 1} q_{I}\) into
  \(q_I \xrightarrow{a, 0<x<1, 1} q_{I}\) and
  \(q_I \xrightarrow{a, x=0, 0} q_{I}\).
\end{example}

A run of a strict \(1\)-IRTA on a word \(u\) can be factored similarly
as explained for a general \(1\)-IRTA, however now, every \(r_{i_j}\)
will be a reset transition: notice that we require each transition to
be guarded using constraints of a special form, either \(x = m\) or
\(m < x < m + 1\) or \(K < x\); therefore, the transition reading
\((t_{i_1}, a_{i_1})\) will necessarily have an equality guard
\(x = m\) forcing a reset, similarly at \(i_2\) and so on.  Therefore,
the sequence \(s_u\) identifies the exact reset points in the word, no
matter which strict \(1\)-IRTA reads it.  The quantity
\(\mathit{c}^K(u)\) gives the value of the clock on reading \(u\) by
any strict \(1\)-IRTA.  This \emph{input-determinism} is a fundamental
property of strict \(1\)-IRTAs that helps in the Myhill-Nerode
characterization that we present in the later sections.

The question now is how expressive are strict \(1\)-IRTAs.  As shown
by the proposition below, every language definable by a \(1\)-IRTA is
also definable by a strict \(1\)-IRTA.  Therefore, we could simply
consider strict \(1\)-IRTAs instead of \(1\)-IRTAs.  Even though a
proof of this equi-expressivity theorem is
known~\cite{bhaveAddingDenseTimedStack2017} we provide one in
appendix~\ref{app:strictness}.

\begin{restatable}[see also Theorem~1
  \cite{bhaveAddingDenseTimedStack2017}]{proposition}{PropRPlus}
  \label{prop:R_plus}
  A language accepted by a (deterministic) \(1\)-IRTA is also accepted
  by a (deterministic) strict \(1\)-IRTA with no greater constant in
  guards.
\end{restatable}
\section{From Equivalences to Automata and Back}
\label{sec:from-equiv-autom}

We start the study of equivalences for languages accepted by integer
reset automata.  Proposition~\ref{prop:R_plus} says that every IRTA
language can be recognized by a strict 1-IRDTAs.  There are two
advantages of strict 1-IRDTAs: there is a single clock; and the value
of the clock on reading the word is simply determined by the word and
not by the automaton that is reading it.  This motivates us to
restrict our attention to equivalences that make use of the quantity
$c^K(u)$, and from which one can construct a strict 1-IRDTA with
states as the equivalence classes.  In order to be able to do so, we
need a good notion of monotonicity (Challenge 3 of the Introduction).

Intuitively, the equivalence should satisfy two conditions whenever
$u$ is equivalent to $v$: (1) when $u$ can elapse time and satisfy a
guard, $v$ should be able to elapse some time and satisfy the same
guard, and (2) all one step extensions of $u$ and $v$, say
$u' = u (t \cdot a)$ and $v' = v (t'\cdot a)$ such that the clock
values $\ck(u')$ and $\ck(v')$ satisfy the same set of guards w.r.t
constant $K$, should be made equivalent.  Each guard in a strict
1-IRDTA represents a $K$-region.  All these remarks lead to the
following definition of $K$-monotonicity.

\begin{definition}[$L$-preserving,
  $K$-monotonic] \label{def:monotonicity} An equivalence
  \(\mathord{\approx}\subseteq \mathord{\mathbb{T}\Sigma^{*}\times
    \mathbb{T}\Sigma^{*}}\) is \(L\)-\emph{preserving} when
  \(u\approx v \implies (u\in L \iff v\in L)\).  Given a constant
  \(K\in \mathbb{N}\),
  \(\mathord{\approx}\subseteq \mathbb{T}\Sigma^{*}\times
  \mathbb{T}\Sigma^{*}\) is \(K\)-\emph{monotonic} when $u \approx v$
  implies:
  \begin{alphaenumerate}
  \item $\ck(u) \equiv^K \ck(v)$, and
  \item $\forall a \in \Sigma$,
    $\forall t, t' \in \mathbb{R}_{\ge 0}:~$
    $\ck(u) + t \equiv^K \ck(v) + t' ~\implies~ u (t \cdot a) \approx
    v (t' \cdot a)$.
  \end{alphaenumerate}
\end{definition}

From a $K$-monotonic equivalence, we can construct a strict 1-IRDTA
whose states are the equivalence classes.  Here is additional
notation.  For a number $t \in \mathbb{R}_{\ge 0}$, we define a clock
constraint $\phi(~[t]_{\equiv^K}~)$ as:
\[ \phi(~[t]_{\equiv^{K}}~) =
  \begin{cases}
    x=t                                   & \quad\text{if } t\leq K \wedge t\in \mathbb{N}\enspace,     \\
    \minor{t} < x \land  x < \minor{t} +1 & \quad\text{if } t\leq K \wedge  t\notin \mathbb{N}\enspace, \\
    K<x & \quad\text{if } K < t \enspace .
  \end{cases} \]

\begin{definition}[From equivalence $\approx$ to strict 1-IRDTA
  $\Aa_\approx$]
  \label{def:Aapprox}
  Let $L \subseteq \tss$, $K \ge 0$, and $\approx$ a $K$-monotonic,
  $L$-preserving equivalence with finite index.  The strict 1-IRDTA
  $\Aa_\approx$ has states $\{ [u]_\approx \mid u \in \tss \}$.  The
  initial state is $[\epsilon]_\approx$.  Final states are
  $\{ [u]_\approx \mid u \in L\}$.  Between two states $[u]_\approx$
  and $[v]_\approx$ there is a transition
  $([u]_\approx, [v]_\approx, a, g, s)$ if there exists
  $t \in \mathbb{R}_{\ge 0}$ such that:
  \begin{itemize}
  \item $ u (t \cdot a) \approx v$, and
  \item $g = \phi(~[c^K(u) + t]_{\equiv^K}~)$, and
  \item $s = 0$ if $c^K(u) + t \in \{0, 1, \dots, K\}$ and $s = 1$
    otherwise.
  \end{itemize}
\end{definition}

We now explain why the above definition does not depend on the
representative picked from an equivalence class.  Suppose
$([u]_\approx, [v]_\approx, a, g, s)$ is a transition.  Let
$t \in \mathbb{R}_{\ge 0}$ be a value which witnesses the transition,
that is, it satisfies the conditions of the above definition.  Pick
another word $u'$ equivalent to $u$, that is, $u \approx u'$.  By
Definition~\ref{def:monotonicity} (a), we have
$\ck(u) \equiv^K \ck(u')$.  Therefore, there exists $t'$ such that
$\ck(u) + t \equiv^K \ck(u') + t'$.  Moreover by (b),
$u' (t' \cdot a) \approx u (t \cdot a)$, and hence
$u' (t' \cdot a) \approx v$.  Therefore, we observe that even if we
had chosen $u'$ instead of $u$, we get a witness $t'$ for the same
transition $([u]_\approx, [v]_\approx, a, g, s)$.

\begin{lemma}\label{lem:monotone-to-automaton}
  Let $\approx$ be an $L$-preserving, $K$-monotonic equivalence with
  finite index.  Then $\Ll(\Aa_\approx) = L$.
\end{lemma}
\begin{proof}
  By induction on the length of the timed words we show that for every
  \(u\in \mathbb{T}\Sigma^{*}\),
  \(([\epsilon]_{\approx},0) \rightsquigarrow^{u} ([u]_{\approx},
  \mathit{c}^{K}(u)) \).
  Let \(u(t\cdot a)\in \mathbb{T}\Sigma^{*}\).  Assume
  \(([\epsilon]_{\approx}, 0) \rightsquigarrow^u
  ([u]_{\approx},\mathit{c}^{K}(u))\).  By definition of
  \(\mathcal{A}_{\approx}\), there is a transition
  \(([u]_{\approx},[u(t\cdot a)]_{\approx}, a,g, s)\) such that
  \(g=\phi([\mathit{c}^{K}(u)+t]_{\equiv^{K}})\) and, $s = 0$ if
  $c^K(u) + t \in \{0, 1, \dots, K\}$ and $s = 1$ otherwise.  Since
  \(\mathit{c}^{K}(u(t\cdot a))=(\mathit{c}^{K}(u)+t)s\) we deduce
  that
  \(([\epsilon]_{\approx},0)
  \rightsquigarrow^{u}([u]_{\approx},\mathit{c}^{K}(u))
  \rightsquigarrow^{(t\cdot a)}([u(t\cdot
  a)]_{\approx},\mathit{c}^{K}(u(t\cdot a)))\).  Finally, since
  \(\mathord{\approx}\) is \(L\)-preserving, a word is in \(L\) if{}f
  \(\mathcal{A}_{\approx}\) accepts it.
\end{proof}

We now look at the reverse question of obtaining a monotonic
equivalence from an automaton.  Given a complete strict 1-IRDTA $\Bb$,
we define an equivalence $\approx_\Bb$ as $u \approx_\Bb v$ if $\Bb$
reaches the same (control) state on reading \(u\) and \(v\) from its
initial configuration.  If $K$ is the maximum constant appearing in
$\Bb$, it is tempting to think that $\approx_\Bb$ is a $K$-monotonic
equivalence.  However $\approx_\Bb$ need not satisfy condition \ma of
Definition~\ref{def:monotonicity}.  For instance, consider a strict
1-IRDTA $\Bb$ which has two self-looping transitions in its initial
state: $q \xra{a, x=0, 0} q$ and $q \xra{a, 0 < x < 1, 1} q$.  Observe
that $(0 \cdot a) \approx_\Bb (0.5 \cdot a)$, but
$c^1((0 \cdot a)) \neq c^1((0.5 \cdot a))$.  Therefore, the
state-based equivalence $\approx_\Bb$ needs to be further refined in
order to satisfy monotonicity.  This leads us to define an equivalence
$\approx_\Bb^K$ as: $u \approx_\Bb^K v$ if $u \approx_\Bb v$ and
$c^K(u) = c^K(v)$.

\begin{lemma}\label{lem:automaton-refined-to-monotone}
  Let $\Bb$ be a complete strict 1-IRDTA with maximum constant $K$.
  The equivalence $\approx_\Bb^K$ is $\Ll(\Bb)$-preserving,
  $K$-monotonic and has finite index.
\end{lemma}
\begin{proof}
  The equivalence $\approx_\Bb^K$ is $\Ll(\Bb)$-preserving because
  equivalent words reach the same state in $\Bb$, thus either both are
  accepted or both are rejected.  It has finite index because the
  number of its equivalence classes is bounded by the number of states
  of \(\mathcal{B}\) multiplied by the number of \(K\) regions.
  Condition \textcolor{lipicsGray}{\sffamily\bfseries\upshape (a)} and
  \textcolor{lipicsGray}{\sffamily\bfseries\upshape (b)} of
  Def.~\ref{def:monotonicity} respectively hold by definition of
  $\approx_\Bb^K$ and since \(\mathcal{B}\) is deterministic.
\end{proof}

The goal of the section was to go from equivalences to automata and
back.  Lemma~\ref{lem:monotone-to-automaton} talks about
equivalence-to-automata.  For the automata-to-equivalence, we needed
to strengthen the state-based equivalence with the region equivalence.
A close look at $\Aa_{\approx}$ of Definition~\ref{def:Aapprox}
reveals whenever $u \approx v$, we also have $\ck(u) \equiv^K \ck(v)$.
So, in the equivalence-to-automata, we get an automaton satisfying a
stronger property.  This motivates us to explicitly highlight a class
of strict 1-IRDTAs where each state can be associated with a unique
region.  For this class, we will be able to go from
equivalence-to-automata-and-back directly.

\begin{definition}[$K$-acceptor]
  A $K$-acceptor $\Bb$ is a complete strict 1-IRDTA with maximum
  constant smaller than or equal to $K$ such that for every
  $u, v \in \tss$, $u \approx_\Bb v$ implies $c^K(u) \equiv^K c^K(v)$.
  Hence every state $q$ of $\Bb$ can be associated to a unique
  $K$-region, denoted, \(\mathit{region}(q)\), i.e., whenever
  \((q_{I},0) \rightsquigarrow^{w} (q, \mathit{c}^{K}(w))\) then
  \(\mathit{c}^{K}(w)\in \mathit{region}(q)\).
\end{definition}

Every strict 1-IRDTA $\Bb$ with maximum constant $K$ can be converted
into a $K$-acceptor by starting with the equivalence $\approx_\Bb^K$
and building $\Aa_{\approx_\Bb^K}$.  Furthermore $K$-monotonic
equivalences characterize $K$-acceptors
(Lemmas~\ref{lem:monotone-to-automaton} and
\ref{lem:automaton-refined-to-monotone}).  Therefore, for the rest of
the document, we will restrict our focus to $K$-acceptors.  As a next
task, we look for the coarsest possible $K$-monotonic equivalence for
a language.  This will give a minimal $K$-acceptor for the language,
which we deem to be the canonical integer reset timed automaton, with
maximum constant $K$, for the language.

\section{A Nerode-style Equivalence}
\label{sec:nerode-style-equiv}

In the previous section, we have established generic conditions
required from an equivalence to construct a $K$-acceptor from it.  In
this section, we will present a concrete such equivalence: given a
language \(L\) definable by a \(1\)-IRTA with a maximum constant
\(K\in \mathbb{N}\) we define a \emph{syntactic equivalence}
\(\mathord{\approx^{L, K}}\subseteq
\mathord{\mathbb{T}\Sigma^{*}\times \mathbb{T}\Sigma^{*}}\).
``Syntactic'' means this equivalence is independent of a specific
representation of \(L\).  We then show that \(\approx^{L, K}\) is the
coarsest \(L\)-preserving and \(K\)-monotonic equivalence.

The idea for defining \(\approx^{L, K}\) is to identify two words
\(u\) and \(u'\) whenever
\(\mathit{c}^{K}(u) \equiv^K \mathit{c}^{K}(u')\) and the residuals
\(u^{-1}L\) and \(u'^{-1}L\) coincide modulo some rescaling
w.r.t. \(\mathit{c}^{K}(u)\) and \(\mathit{c}^{K}(u')\).  We start
with examples to give some intuition behind the rescaling function.

\subparagraph{Examples.}  Consider an automaton segment
$q_0 \xra{a, 0 < x < 1, 1} q_1 \xra{b, x = 1, 0} q_2 \xra{c, 0 < x <
  1, 1} q_3$, with $q_3$ being an accepting state.  The words
$u = (0.2 \cdot a)$ and $v = (0.6 \cdot a)$ both go to state $q_1$.
Let us assume that all other transitions go to a sink state, and also
that the maximum constant $K = 1$.  The language $L$ accepted is
$\{ (t_1 \cdot a) (t_2 \cdot b) (t_3 \cdot c) \}$ where $0 < t_1 < 1$,
$t_1 + t_2 = 1$ and $0 < t_3 < 1$.  Moreover,
$u^{-1}L = \{ (0.8 \cdot b) (t_3 \cdot c) \mid 0 < t_3 < 1 \}$ and
$v^{-1}L = \{(0.4 \cdot b) (t_3 \cdot c) \mid 0 < t_3 <1 \}$.  Here we
want to somehow ``equate'' the residual languages $u^{-1}L$ and
$v^{-1} L$.  The idea is to define a bijection between these two sets
$u^{-1}L$ and $v^{-1}L$.  In this case, the bijection maps
$(0.8 \cdot b) ( t_3 \cdot c)$ to $(0.4 \cdot b) (t_3 \cdot c)$ for
every $t_3$.  Observe that given $u, v$ the bijection depends on the
values $\ck(u)$ and $\ck(v)$, which in this example are $0.2$ and
$0.6$ respectively.

Here is another example.  Let $u_1, v_1$ be words with
$\ck(u_1) = 0.2$ and $\ck(v_1) = 0.6$.  Let
$u_1^{-1}L = \{ (t_1 \cdot a) (t_2 \cdot b) (t_3 \cdot c) (t_4 \cdot
d) \mid \ck(u_1) + t_1 + t_2 + t_3 = 1 \}$ and
$v_1^{-1}L = \{ (t'_1 \cdot a) (t'_2 \cdot b) (t'_3 \cdot c) (t'_4
\cdot d) \mid \ck(v_1) + t'_1 + t'_2 + t'_3 = 1 \}$.  The bijection in
this case is more complicated than the previous example.  The idea is
to first start with $\ck(u_1) = 0.2$, $\ck(v_1) = 0.6$ and consider a
bijection $f$ of the open unit interval $(0, 1)$ that maps the
intervals $(0, 0.8]$ and $(0.8, 1)$ to $(0, 0.4]$ and $(0.4, 1)$
respectively.  This bijection is essentially a rescaling of the
intervals $(0, 1-\ck(u_1)]$ and $(1 - \ck(u_1), 1)$ into the intervals
$(0, 1 - \ck(v_1)])$ and $(1 - \ck(v_1), 1)$.  We now pick the first
letters in the residual languages $u_1^{-1}L$ and $v^{-1}L$ and create
a mapping: $(t_1 \cdot a) \mapsto (f(t_1) \cdot a)$.  Now we consider
$\ck(u_1 (t_1 \cdot a))$ and $\ck(v_1 (f(t_1) \cdot a))$ in the place
of $\ck(u_1), \ck(v_1)$, and continue the mapping process one letter
at a time.

\subparagraph{Rescaling function.}  We will now formalize this idea.
We will start with bijections of the open unit interval.

Let $\lm, \lm' \in (0, 1)$ be arbitrary real values.  Define a
bijection $f_{\lm \to \lm'}: (0, 1) \to (0, 1)$ that scales the
$(0, \lm]$ interval to $(0, \lm']$ and the $(\lm, 1)$ interval to
$(\lm', 1)$:
\begin{align*}
  f_{\lm \to \lm'}(t) = \begin{cases}
    \left( \dfrac{\lm'}{\lm} \right) t             & \text{ for } 0 < t \le \lm \\
    \lm' + \dfrac{(1 - \lm')}{(1 - \lm)} (t - \lm) & \text{ for } \lm < t
    < 1
  \end{cases}
\end{align*}

Now, consider $x, x' \in \mathbb{R}_{\ge 0}$ such that
$x \equiv^K x'$.  We define a \emph{length-preserving} bijection
$\tau_{x \to x'}: \mathbb{T} \to \mathbb{T}$ inductively as follows:
for the empty sequence $\epsilon$, we define
$\tau_{x \to x'}(\epsilon) = \epsilon$; for a timestamp
$d \in \mathbb{T}$ and a $t \in \mathbb{R}_{\ge 0}$,
$\tau_{x \to x'}(d t) = d' t'$ where $d' = \tau_{x \to x'}(d)$ and
$t'$ is obtained as follows: let $y = \ck(xd)$ and $y' = \ck(x'd')$.
If $y, y' \in \mathbb{N}$ or $y, y' >K$, then define $t' = t$.  Else,
define $\lfloor t' \rfloor = \lfloor t \rfloor$ and
$\{t'\} = f_{(1-\{y\}) \to (1-\{y'\})}(\{t\})$.

Here is an additional notation, before we describe some properties of
the rescaling function.  For $x_1, x_2, x_3 \in \mathbb{R}_{\ge 0}$
with $x_1 \equiv^K x_2 \equiv^K x_3$, we denote the composed function
$(\tau_{x_2 \to x_3}) \circ (\tau_{x_1 \to x_2})$ as
$\tau_{x_1 \to x_2 \to x_3}$.  So,
$\tau_{x_1 \to x_2 \to x_3}(t) = \tau_{x_2 \to x_3}(\tau_{x_1 \to
  x_2}(t))$.

\begin{lemma}\label{lem:bijection-properties}
  The bijection $\tau_{x \to x'}$ satisfies the following properties:
  \begin{enumerate}
  \item for an arbitrary timestamp $t_1 t_2 \dots t_n \in \mathbb{T}$,
    if $\tau_{x \to x'}(t_1 t_2 \dots t_n) = t'_1 t'_2 \dots t'_n$
    then, we have
    $\ck(x + t_1 + \dots + t_{n-1})+t_{n} \equiv^K \ck(x' + t'_1 +
    \dots + t'_{n-1})+t'_{n}$,
  \item $\tau_{x \to x'}^{-1}$ is identical to $\tau_{x' \to x}$,
  \item $\tau_{x_1 \to x_2 \to x_3}$ is identical to
    $\tau_{x_1 \to x_3}$.
  \end{enumerate}
\end{lemma}

The rescaling function $\tau_{x \to x'}$ can be naturally extended to
timed words:
$\tau_{x \to x'}( (t_1 \cdot a_1) (t_2 \cdot a_2) \dots (t_n \cdot
a_n) ) = (t'_1 \cdot a_1) (t'_2 \cdot a_2) \dots (t'_n \cdot a_n)$
where $t'_1 t'_2 \dots t'_n = \tau_{x \to x'}(t_1 t_2 \dots t_n)$.
The next observation follows from Property \mone of
Lemma~\ref{lem:bijection-properties}.

\begin{lemma}\label{lem:rescaling-function-and-autamaton-runs}
  Let $\Bb$ be a $K$-acceptor, and $q$ a control state of $\Bb$.  Let
  $x, x' \in \mathbb{R}_{\ge 0}$ such that $x \equiv^K x'$.  Then, for
  every timed word $w$: \ $w \in \mathcal{L}(q, x)$ if{}f
  $\tau_{x \to x'}(w) \in \mathcal{L}(q, x')$.
\end{lemma}

\subparagraph*{Syntactic equivalence.}  For timed words
$u, v \in \tss$ such that $\ck(u) \equiv^K \ck(v)$, we write
$\tau_{u \to v}$ for the bijection $\tau_{\ck(u) \to \ck(v)}$.  We now
present the main equivalence.

\begin{definition}[Equivalence
  $\approx^{L,K}$] \label{def:main-equivalence} Let $L$ be a timed
  language and $K$ a natural number.  We say \(u\approx^{L, K} v\) if
  $\ck(u) \equiv^K \ck(v)$ and \(\tau_{u \to v}(u^{-1}L)=v^{-1}L\).
\end{definition}

Note that the equivalence \(\approx^{L, K} \) is \(L\)-preserving.
Assume \(u\approx^{L, K} v\).  We have
\(u\in L\iff \epsilon \in u^{-1}L\) and
\(\epsilon \in u^{-1}L \iff \epsilon\in v^{-1}L\) since
\(\tau_{u \to v}(\epsilon)=\epsilon\) by definition.  Thus,
\(u\in L \iff v\in L\).

\begin{proposition}
  \label{prop:syntactic}
  When \(L\subseteq \mathbb{T}\Sigma^{*}\) is definable by a
  $K$-acceptor, then \(\approx^{L,K}\) has finite index and is
  \(K\)-monotonic.  Moreover, \(\approx^{L,K}\) is the coarsest
  \(K\)-monotonic and \(L\)-preserving equivalence.
\end{proposition}
\begin{proof}
  We start by showing that $\approx^{L, K}$ is $K$-monotonic.
  Condition \ma of Definition~\ref{def:monotonicity} holds by
  definition.  We move to \mb.  Let $t_u, t_v \in \mathbb{R}_{\ge 0}$
  s.t.  $\ck(u) + t_u \equiv^K \ck(v) + t_v$. Let:
  \begin{align}
    u_1 = u (t_u \cdot a) \quad v_1 & = v (t_v \cdot a)  \nonumber \\
    \text{ To show:} \quad \tau_{u_1 \to v_1}(u_1^{-1}
    L)                              & = v_1^{-1}L \label{eq:1}\end{align} Let $t'_v = \tau_{u \to v}(t_u)$ and $v_2 = v (t'_v \cdot a)$.
  We
  will prove \ref{eq:1} using these intermediate claims:
  \begin{claim}\label{claim:u1-to-v2}
    $\tau_{u_1 \to v_2}(u_1^{-1}L) = v_2^{-1}L$
  \end{claim}
  \begin{claim}\label{claim:v2-to-v1}
    $\tau_{v_2 \to v_1}(v_2^{-1} L) = v_1^{-1}L$ \end{claim} Hence:
  $\tau_{u_1 \to v_2 \to v_1}(u_1^{-1}L) = v_1^{-1}L$.  By
  Lemma~\ref{lem:bijection-properties}, we conclude
  $\tau_{u_1 \to v_1}(u_1^{-1}L) = v_1^{-1}L$.

  \emph{Proof of Claim~\ref{claim:u1-to-v2}.}  Let $w \in \tss$.  By
  definition of $u_1$, we have:
  \begin{align}\label{eq:2}
    w \in u_1^{-1}
    L \quad \text{ if{}f } \quad (t_u \cdot a) w \in u^{-1}L\end{align} Since $u \approx^{L,K} v$, we know that $\tau_{u \to v}(u^{-1}L) = v^{-1}L$.
  Therefore:
  \begin{align}\label{eq:3}
    (t_u \cdot a) w \in u^{-1}L \quad \text{ if{}f } \quad \tau_{u
    \to v}((t_u \cdot a) w) \in v^{-1}L
  \end{align}
  By the way we have constructed the rescaling function, we have
  \begin{align}\label{eq:4}
    \tau_{u \to v}( (t_u \cdot a) w) = (t'_v \cdot a) \tau_{u_1 \to
    v_2}(w)
  \end{align}
  Finally, by definition of $v_2$:
  \begin{align}\label{eq:5}
    (t'_v \cdot a) \tau_{u_1 \to v_2}(w) \in v^{-1}
    L \quad \text{ if{}f } \quad \tau_{u_1 \to v_2}(w) \in v_2^{-1}L\end{align} From \eqref{eq:2}, \eqref{eq:3}, \eqref{eq:4} and \eqref{eq:5}, we conclude $w \in u_1^{-1}L$ if{}f $\tau_{u_1 \to v_2}(w) \in v_1^{-1}L$ for an arbitrary timed word $w$.
  This proves the claim.

  \emph{Proof of Claim~\ref{claim:v2-to-v1}}.  Let $\Bb$ be a
  $K$-acceptor recognizing $L$, and let $q$ be the control state
  reached by $\Bb$ on reading word $v$.  We first claim that:
  \begin{align}\label{eq:6}
    \ck(v) + t_v \equiv^K \ck(v) + t_v'
  \end{align}
  This is because, by Lemma~\ref{lem:bijection-properties}, we have
  $\ck(u) + t_u \equiv^K \ck(v) + t'_v$, and by assumption, we have
  $\ck(u) + t_u \equiv^K \ck(v) + t_v$.  Since $\equiv^K$ is
  transitive, \eqref{eq:6} follows.

  The observation made in \eqref{eq:6} implies on elapsing $t_v$ or
  $t'_v$ from $v$, the same outgoing transition is enabled, as $\Bb$
  is deterministic.  Therefore $v_1 = v (t_v \cdot a)$ and
  $v_2 = v (t'_v \cdot a)$ reach the same control state $q'$.  Hence,
  \eqref{eq:6} can be read as $\ck(v_1) \equiv^K \ck(v_2)$.  By
  Lemma~\ref{lem:rescaling-function-and-autamaton-runs}, for any timed
  word $w$, we have $w \in v_2^{-1}L$ if{}f
  $\tau_{v_2 \to v_1}(w) \in v_1^{-1}L$.  The claim follows.

  We will now show that $\approx^{L,K}$ has finite index and is also
  the coarsest $K$-monotonic, $L$-preserving equivalence.  Let $\Bb$
  be a $K$-acceptor for $L$.  Recall that $u \approx^\Bb v$ if $\Bb$
  reaches the same control state on reading $u$ and $v$.  We will
  show:
  \begin{align}\label{eq:7}
    u \approx^\Bb v \quad \text{ implies } \quad u \approx^{L, K} v
  \end{align}
  This will immediately prove that $\approx^{L,K}$ has finite index.
  Secondly, for any $K$-monotonic, $L$-preserving equivalence
  $\approx$, we can build a $K$-acceptor $\Aa_{\approx}$
  (Lemma~\ref{lem:monotone-to-automaton}) whose equivalence is
  identical to $\approx$.  Thus, \eqref{eq:7} also shows that
  $\approx^{L,K}$ is the coarsest such equivalence.

  Proof of \eqref{eq:7} is as follows.  Let $q$ be the control state
  reached by $u$ and $v$ in $\Bb$, and let $x = \ck(u), x' = \ck(v)$.
  Since $\Bb$ is a $K$-acceptor, we also have $x \equiv^K x'$.  From
  Lemma~\ref{lem:rescaling-function-and-autamaton-runs},
  $\tau_{x \to x'}(\Ll(q, x)) = \Ll(q, x')$.  Hence, we deduce:
  $\tau_{u \to v}(u^{-1}L) = v^{-1}L$.  This proves
  $u \approx^{L,K} v$.
\end{proof}

The above proposition leads to the following Myhill-Nerode style
characterization for timed languages recognized by IRTA.

\begin{theorem}
  \label{thm:good-equivalence-0}
  \hspace{0pt}%
  \begin{alphaenumerate}
  \item \(L\subseteq \mathbb{T}\Sigma^{*}\) is a language definable by
    an IRTA if and only if there is a constant \(K\in \mathbb{N}\)
    such that \(\approx^{L,K}\) is \(K\)-monotonic and has finite
    index.
  \item \(\approx^{L,K}\) is coarser than any \(K\)-monotonic and
    \(L\)-preserving equivalence.
  \end{alphaenumerate}
\end{theorem}

The conversion from a general IRTA to a strict 1-IRTA does not
increase the constant.  Similarly, a strict 1-IRTA with maximum
constant $K$ can be converted to a $K$-acceptor.  Therefore, the
overall conversion from an IRTA to an acceptor preserves the constant.
From the second part of Theorem~\ref{thm:good-equivalence-0}, we
deduce that for a language $L$ that is recognized by an IRTA with
constant $K$, the $K$-acceptor $\Aa_{\approx^{L,K}}$ built from the
equivalence $\approx^{L,K}$ has the least number of states among all
$K$-acceptors recognizing $L$.  It is therefore a minimal automaton
among all $K$-acceptors for $L$.  We now present some examples that
apply this Myhill-Nerode characterization. 

\begin{example}
  Consider the language \(L=\{(x\cdot a) \mid x\in \mathbb{N}\}\).
  The right-and side of of Theorem~\ref{thm:good-equivalence-0}
  \textcolor{lipicsGray}{\sffamily\bfseries\upshape (a)} does not
  hold.  Indeed, there is no \(K\in \mathbb{N}\) such that
  \(\approx^{L,K}\) verifies the \(K\)-monotonicity
  \textcolor{lipicsGray}{\sffamily\bfseries\upshape (b)}: For every
  \(K \in \mathbb{N}\) we have
  \(\mathit{c}^{K}(\epsilon)+K+1\equiv^{K}\mathit{c}^{K}(\epsilon)+K+1.1\).
  Since \((K+1\cdot a)\in L\) and \((K+1.1\cdot a)\notin L\), by
  \(L\)-preservation, \((K+1\cdot a)\not\approx^{L,K}(K+1.1\cdot a)\).
  By Theorem~\ref{thm:good-equivalence-0}, \(L\) is not \(1\)-IRTA
  definable.
\end{example}
\begin{example}
  Consider the language \(L=\{(x\cdot a) (1\cdot b)\mid 0<x<1\}\).
  This language is accepted by a \(1\)-TA that reads \(a\) on a guard
  \(0 < x < 1\), resets the clock \(x\) and reads \(b\) at \(x = 1\).
  This is clearly not an IRTA.  We will once again see that
  \(K\)-monotonicity holds for no \(K\).  For \(u=(0.2\cdot a)\),
  \(\mathit{c}^{K}(u)+1\equiv^{K}\mathit{c}^{K}(u)+1.1\) holds for
  every constant \(K\in \mathbb{N}\).  Since \(u(1\cdot b)\in L\) and
  \(u(1.1\cdot b) \notin L\), we have
  \(u(1\cdot b)\not\approx^{L,K}u(1.1\cdot b)\).  Thus, there is no
  \(K\) such that \(\approx^{L,K}\) verifies \(K\)-monotonicity
  \textcolor{lipicsGray}{\sffamily\bfseries\upshape (b)}, hence \(L\)
  is not \(1\)-IRTA definable by Theorem~\ref{thm:good-equivalence-0}.
\end{example}
\begin{example}
  Consider the language
  \(M=\{u \in \mathbb{T}\Sigma^{*}\mid \mathit{c}^{1}(u)=0\}\) with
  alphabet \(\Sigma=\{a\}\) which has the following residual
  languages.  For \(u\in \mathbb{T}\Sigma^{*}\),
  \begin{align*}
    u^{-1}
    M & =M & \text{when \( \mathit{c}^{1}(u)=0\),} &  & u^{-1}M & =\{v\in \mathbb{T}\Sigma^{*}\mid \sigma(uv)=1\} \quad\text{ when \(0<\mathit{c}^{1}(u)<1\),} \\ u^{-1}M & =\emptyset & \text{when \(1<\mathit{c}^{1}(u)\).
                                                                                                                                                                                         }
  \end{align*}
  The equivalence classes for \(\approx^{M,1}\) are the following:
  \begin{align*}
    [\epsilon]_{\approx^{M,1}} & =M, & [({\textstyle\frac{3}{2}} \cdot a)]_{\approx^{M,1}} & =\{u\in \mathbb{T}\Sigma^{*}\mid 1<\mathit{c}^{1}(u)\}, & [({\textstyle\frac{1}{2}} \cdot a)]_{\approx^{M,1}} & =\{u\in \mathbb{T}\Sigma^{*}\mid0<\mathit{c}^{1}(u)<1\}.
  \end{align*}
  The acceptor \(\mathcal{A}_{\approx^{M,1}}\) is depicted in
  \figurename~\ref{fig:ta_R_int_not_I} where
  \([\epsilon]_{\approx^{M,1}}\) is the initial (and final) state,
  \([(\frac{3}{2}\cdot a)]_{\approx^{M,1}}\) the sink state, and
  \([(\frac{1}{2}\cdot a)]_{\approx^{M,1}}\) the rightmost state.
\end{example}

\begin{example}
	Consider the language \(L=\{(0\cdot a)^{n}(0\cdot b)^{n}\mid n \in \mathbb{N}\}\).
	There is no \(K\) such that \(\approx^{L,K}\) has finite index, although \(\approx^{L,0}\) is \(0\)-monotonic: for \(u\approx^{L,0}v\) with \(0 < \mathit{c}^0(u)\equiv^{0}\mathit{c}^0(v)\), no extension of \(u\) or $v$ belongs to \(L\) and hence monotonicity \textcolor{lipicsGray}{\sffamily\bfseries\upshape (b)} holds; pick \(u\) and $v$ with \(\mathit{c}^0(u) = \mathit{c}^0(v)=0\), and let \(t, t'\) such that \(\mathit{c}^0(u) + t \equiv^{0} \mathit{c}^0(v) + t'\).
	If \(t = 0\) then \(t' = 0\) and monotonicity holds.
	Suppose \(0 < t, t'\).
	Then \(u (t \cdot \alpha) \notin L\) and \(u (t' \cdot \alpha) \notin L\) for any letter \(\alpha\).
	Once again, \textcolor{lipicsGray}{\sffamily\bfseries\upshape (b)} holds.

	We now argue the infinite index, similar to the case of untimed languages.
	For distinct integers \(n\) and \(m\) we have \(((0\cdot a)^{n})^{-1}L \neq ((0\cdot a)^{m})^{-1}L\).
	Since \(\tau_{0 \to 0}\) is the identity we have \(\tau_{0 \to 0}(((0\cdot a)^{n})^{-1}L)\neq ((0\cdot a)^{m})^{-1}L\).
	Thus, \((0\cdot a)^{n} \not \approx^{L,K}(0\cdot a)^{m}\).
\end{example}
\begin{example}
	\label{example:canonical}
	Consider the language \(L=\{u \in \mathbb{T}\Sigma^{*}\mid t_1+\dots +t_n=1\}\) given in \figurename~\ref{fig:ta_I}.
	Given a timed word \(u=(t_1\cdot a_{1})\dots (t_n\cdot a_{n})\in \mathbb{T}\Sigma^{*}\) we denote by \(\sigma(u)\) the sum \(t_1+\dots +t_n\) of its timestamps.
	Also we have that \(\sigma(\epsilon) = 0\).
	The residual languages of \(L\) are the following.
	For \(u \in \mathbb{T}\Sigma^{*}\),
	\begin{align*}
		u^{-1}
		L & =\emptyset & \text{when \(1<\sigma(u)\),} &  & u^{-1}L & =\{v\in \mathbb{T}\Sigma^{*}\mid \sigma(v)=0\} &  & \text{when \(\sigma(u)=1\),} \\ u^{-1}L & =L & \text{when \( \sigma(u)=0\),} &  & u^{-1}L & =\{v\in \mathbb{T}\Sigma^{*}\mid \sigma(u)+ \sigma(v)=1\} &  & \text{when \(0<\sigma(u)<1\).
			   }
	\end{align*}
	The equivalence
	classes for \(\approx^{L,1}\) are:
	\begin{align*}
		[(1\cdot a)]_{\approx^{L,1}} & =L,                                                     & [(1\cdot a)({\textstyle\frac{1}{2}}\cdot a)]_{\approx^{L,1}} & =\{u\in \mathbb{T}\Sigma^{*}\mid0<\mathit{c}^{1}(u)<1\wedge 1< \sigma(u)\}, \\
		[\epsilon]_{\approx^{L,1}}   & =\{u\in \mathbb{T}\Sigma^{*}\mid \sigma(u)=0\},         & [({\textstyle\frac{1}{2}}\cdot a)]_{\approx^{L,1}}           & =\{u\in \mathbb{T}\Sigma^{*}\mid0<\mathit{c}^{1}(u), \sigma(u)<1\},         \\
		[(2\cdot a)]_{\approx^{L,1}} & =\{u\in \mathbb{T}\Sigma^{*}\mid 1<\mathit{c}^{1}(u)\}, & [(1\cdot a)(1\cdot a)]_{\approx^{L,1}}                       & =\{u\in \mathbb{T}\Sigma^{*}\mid \mathit{c}^{1}(u)=0 \wedge 1<\sigma(u)\}.
	\end{align*}
	The acceptor \(\mathcal{A}_{\approx^{L,1}}\) is depicted in \figurename~\ref{fig:ta_L_R} and has six states, whereas the (strict) \(1\)-IRDTA given in Example~\ref{example:strict} for \(L\) only has two.
\end{example}
\begin{figure}
	\centering
	\begin{tikzpicture}[shorten >=1pt,node distance=6cm,on grid,auto,initial text={}]
		\node[state,initial]    (epsilon) at (0,0)  {\(\epsilon\)};
		\node[state]            (s) at (4,-2.3)     {\((\frac{1}{2}\cdot a)\)};
		\node[state, accepting] (q) at (7.5,0)      {\((1\cdot a)\)};
		\node[state]            (2) at (4,0)        {\((2\cdot a)\)};
		\node[state]            (z) at (11,-2.3)    {\((1{\cdot} a)(\frac{1}{2}{\cdot} a)\)};
		\node[state]            (r) at (11,2.3)     {\((1{\cdot} a)(1{\cdot} a)\)};
		\path[->]
			(epsilon) edge [loop above]         node                             {\(a,x=0,0\)} ()
			(epsilon) edge node [below, sloped]             {\(a,0<x<1,1\)} (s)
			(epsilon) edge [bend left=40]       node [above]                     {\(a,x=1,0\)} (q)
			(epsilon) edge node                             {\(a,1<x,1\)} (2)
			(s)       edge [loop left]          node [below left]                {\(a,0<x<1,1\)} ()
			(s)       edge node [sloped, near start]        {\(a,x=1,0\)} (q)
			(s)       edge node                             {\(a,1<x,1\)} (2)
			(q)       edge [loop right]         node [right]                     {\(a,x=0,0\)} ()
			(q)       edge [right]              node [sloped, above]             {\(a,0<x<1,1\)} (z)
			(q)       edge [right]              node [sloped, above]             {\(a,x=1,0\)} (r)
			(q)       edge [above]              node                             {\(a,1<x,1\)} (2)
			(2)       edge [loop above]         node [left, near start]          {\(a,1<x,1\)} ()
			(r)       edge [loop right]         node [above, near start]         {\(a,x=0,0\)} ()
			(r)       edge [bend left =50]      node [sloped]                    {\(a,0<x<1,1\)} (z)
			(r)       edge [loop left]          node [above, near end]           {\(a,x=1,0\)} ()
			(r)       edge [above]              node [sloped]                    {\(a,1<x,1\)} (2)
			(z)       edge [loop right]         node [right=.2cm,below,near end] {\(a,0<x<1,1\)} ()
			(z)       edge node [sloped, below]             {\(a,x=1,0\)} (r)
			(z)       edge [above]              node [sloped, near start]        {\(a,1<x,1\)} (2);
	\end{tikzpicture}
	\caption{A strict \(1\)-IRDTA \(\mathcal{A}_{\approx^{L,1}}\) accepting \(L=\{u \in \mathbb{T}\Sigma^{*}\mid \sigma(u)=1\}\).}
	\label{fig:ta_L_R}
      \end{figure}
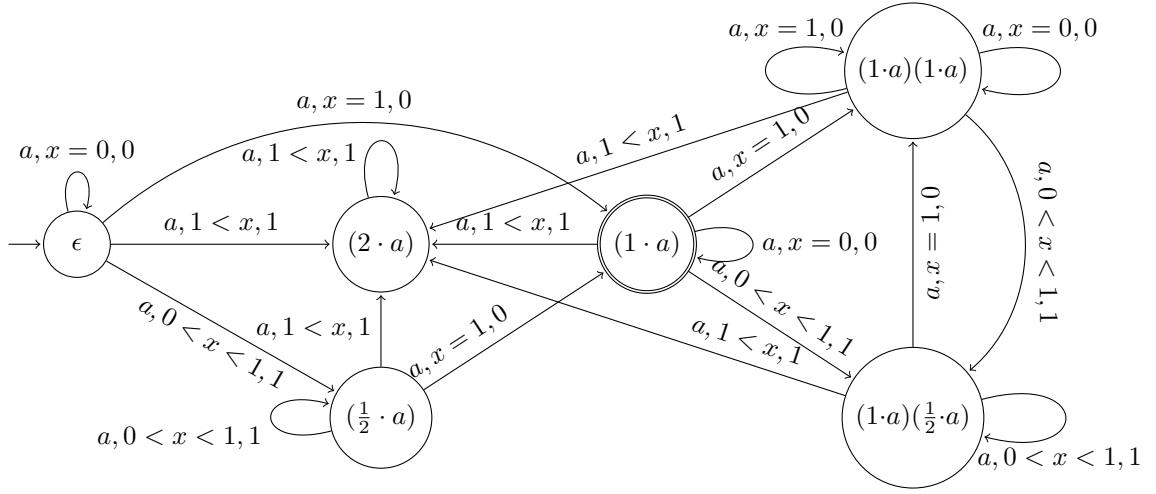

\section{The Canonical Form}
\label{sec:canonical-form}

Sections~\ref{sec:from-equiv-autom} and \ref{sec:nerode-style-equiv}
developed the idea of a canonical $K$-acceptor for a language
recognized by an IRTA with constant $K$.  In this section, we will
further study this canonical form.  We present a crucial property that
will help effectively compute the canonical form, and also apply an
Angluin-style $L^*$ algorithm.

\begin{definition}[Half-integral words]
  We call a timed word
  $(t_1 \cdot a_1) (t_2 \cdot a_2) \dots (t_n \cdot a_n)$ to be a
  \emph{half-integral} if for each $1 \le i \le n$, the fractional
  part $\{t_i\}$ is either $0$ or $\frac{1}{2}$: in other words, each
  delay is either an integer or a rational with fractional value
  $\frac{1}{2}$.  Also, the empty word \(\epsilon\) is a half-integral
  word.
\end{definition}

\begin{definition}[Small half-integral words]
  Let $K \in \mathbb{N}$.  A half-integral word
  $(t_1 \cdot a_1) (t_2 \cdot a_2) \dots (t_n \cdot a_n)$ is said to
  be \emph{small} w.r.t $K$ if $t_i < K + 1$ for all $1 \le i \le n$.
\end{definition}
For a finite alphabet $\Sigma$ and $K \in \mathbb{N}$, let
$\Sigma_K := \{0, \frac{1}{2}, 1, \dots, K - \frac{1}{2}, K, K +
\frac{1}{2}\} \times \Sigma$.  Every small integral word is therefore
in $\Sigma_K^*$.  The next lemma is a generalization of the following
statement: for every timed word $u$, there is a small half-integral
timed word $w$ such that every $K$-acceptor reaches the same state on
reading both $u$ and $w$.

\begin{restatable}{lemma}{halfIntegral}
  \label{lem:prelim}
  Let $u_0$ be a half-integral word which is small w.r.t. $K$.  For
  every timed word $u$, there is a half-integral word $w$ such that
  $u_0w$ is small w.r.t $K$ and every $K$-acceptor reaches the same
  state on reading $u_0u$ or $u_0w$.
\end{restatable}
This allows us to identify the canonical equivalence $\approx^{L, K}$
using small integral words.

\begin{restatable}{proposition}{canonicalHalfIntegral}
  \label{prop:canonical_form}
  Let $\Bb$ be a $K$-acceptor for a language $L$.  Then, the
  equivalence $\approx^\Bb$ coincides with $\approx^{L,K}$ if{}f for
  all half-integral words $u, v \in (\Sigma_K)^*$: $u \approx^\Bb v$
  if{}f $u \approx^{L, K} v$.
\end{restatable}
We make use of Proposition~\ref{prop:canonical_form} to compute the
canonical form.

\subsection{Computing the Canonical Form}

Given a $K$-acceptor $\Bb$, we can minimize it using an algorithm
which is similar to the standard DFA minimization which proceeds by
computing a sequence of equivalence relations on the states.

\begin{itemize}
\item Equivalence $\sim^0$: for a pair of states $p, q$ of $\Bb$
  define $p \sim^0 q$ if $region(p) = region(q)$ and either both are
  accepting states, or both are non-accepting states.
\item Suppose we have computed the equivalence $\sim^i$ for some
  $i \in \mathbb{N}$.  For a pair of states $p, q$, define
  $p \sim^{i+1} q$ if $p \sim^{i} q$ and for every letter
  $(t, a) \in \Sigma_K$, the outgoing transitions
  $(p, p', a, \phi( [t]_{\equiv^K}), s)$ and
  $(q, q', a, \phi( [t]_{\equiv^K} ), s)$ in $\Bb$ satisfy
  $p' \sim^{i} q'$.

\item Stop when $\sim^{i+1}$ equals $\sim^{i}$.
\end{itemize}

The next lemma is a simple consequence of the definition of $\sim^i$
and by an induction on the number of iterations $i$.
\begin{lemma}\label{lem:minimization-iteration-invariant}
  Let $p, q$ be such that $region(p) = region(q)$.  Suppose
  $p \not \sim^i q$.  Then there exists a word $z$ of length at most
  $i$ such that $\delta^*_\Bb(p, z)$ is accepting whereas
  $\delta^*_\Bb(q, z)$ is not.
\end{lemma}

We define the quotient of \( \mathcal{B} \) by \({\sim^i} \) as the
$K$-acceptor whose states are the equivalence classes for $\sim^i$.
There is a transition
$([q]_{{\sim^i}}, [p]_{{\sim^i}},a, \phi([t]_{\equiv^K}),s)$ if there
is $q'\in [q]_{{\sim^i}}$ and $p'\in [p]_{{\sim^i}}$ such that
$(q', p',a, \phi([t]_{\equiv^K},s)$ is a transition of $\mathcal{B}$.
The initial state is the class of the initial state of $\mathcal{B}$
and the final states are the classes that have a non empty
intersection with the set of final states of $\mathcal{B}$.  For
$i\geq 1$ the quotient of \( \mathcal{B} \) by \({\sim^i} \) is an
acceptor for $L(\mathcal{B})$.

Suppose we reach a fixpoint at $m \in \mathbb{N}$.  The quotient of
\( \mathcal{B} \) by $\sim^m$ gives the canonical automaton
$\Aa_{\approx^{L,K}}$.  Suppose the quotient does not induce the
canonical equivalence.  By Proposition~\ref{prop:canonical_form} there
are two words $u, v\in (\Sigma_K)^* $ such that $u$ and $v$ go to a
different state, but $u \approx^{L, K} v$.  Consider the iteration $i$
when the two states were made non-equivalent.  There is a small
half-integral word $z$ of length $i$ which distinguishes $u$ and $v$
by Lemma~\ref{lem:minimization-iteration-invariant} -- a contradiction
to $u \approx^{L,K} v$.  Let us say $z \in u^{-1} L$, since $u, v$ are
half-integral, $\tau_{u \to v}$ is simply the identity function.
Since $z \not \in v^{-1}L$, we deduce that
$\tau_{u \to v}(u^{-1} L) \neq v^{-1}L$, hence
$u \not \approx^{L,K} v$.

\subsection{Learning the Canonical Form}
Our learning algorithm closely follows Angluin's $L^*$ approach, so we
assume familiarity with it and provide a brief example of its
adaptation to the IRTA setting.  Detailed definitions, proof of
correctness, and a complete example are in Appendix~\ref{ap:angluin}.

We assume that the Learner is aware of the maximum constant $K$ for
the unknown language $L$.  The Learner's goal is to identify the
equivalence classes of $\approx^{L,K}$ using small half-integral
words, from $\Sigma_K^*$.  Correspondingly, the rows and columns in an
observation table are words in $\Sigma_K^*$.  In the $L^*$ algorithm,
each row of the observation table corresponds to an identified state.
Two identical rows correspond to the same state.  In order to make a
similar conclusion, we add a column to the observation table, that
maintains $c^K(u)$ for every string $u$ of a row.  There is one
detail: once $c^K(u)$ goes beyond $K$, we want to store it as a single
entity $\top$.  We define $\ck_{\top}(u)$ to be equal to $\ck(u)$ when
this value is bellow $K$ and equal to $\top$ otherwise.
\begin{restatable}{lemma}{forLearning}
  \label{lemma:for_learning}
  Let $L$ be a timed language recognized by a $K$-acceptor, and let
  $u, v\in (\Sigma_K)^*$.  Then $u \approx^{L, K} v$ if{}f
  $\ck_{\top}(u)=\ck_{\top}(v)$ and for all words
  $z \in (\Sigma_K)^*$, we have $uz \in L$ if{}f $vz \in L$.
\end{restatable}

An observation table labels its rows and columns with words in
\((\Sigma_{K})^{*}\).  The row words form a prefix-closed set, and the
column words form a suffix-closed set, as in Angluin.  Table
\ref{tab:learning} shows three observation tables.  The lower part of
these tables includes one-letter extensions of the row words, and
their $\ck_{\top}$ values are shown in the extra column in red.

\begin{table}
  \centering
  \caption{A run of $L^*$ for IRTA}
  \label{tab:learning}
  \begin{tabular}{|c|@{\hspace{2pt}}c@{\hspace{2pt}}|c|}
    \hline
    \(\mathcal{T}_{0}\)        & \(\epsilon\) & \(\color{red}{\mathit{c}_{\top}^{1}(u)}\) \\
    \hline
    \(\epsilon\)               & \(0\)        & \(\color{red}{0}\)                        \\
    \hline
    \((0\cdot a)\)             & \(0\)        & \(\color{red}{0}\)                        \\
    \((\frac{1}{2}\cdot a)\)   & \(0\)        & \(\color{red}{\frac{1}{2}}\)              \\
    \((1\cdot a)\)             & \(1\)        & \(\color{red}{0}\)                        \\
    \((1+\frac{1}{2}\cdot a)\) & \(0\)        & \(\color{red}{\top}\)                     \\
    \hline
  \end{tabular}
  \begin{tabular}{|c|@{\hspace{2pt}}c@{\hspace{2pt}}|c|}
    \hline
    \(\mathcal{T}_{1}\)                            & \(\epsilon\) & \(\color{red}{\mathit{c}_{\top}^{1}(u)}\) \\
    \hline
    \(\epsilon\)                                   & \(0\)        & \(\color{red}{0}\)                        \\
    \((\frac{1}{2}\cdot a)\)                       & \(0\)        & \(\color{red}{\frac{1}{2}}\)              \\
    \((1\cdot a)\)                                 & \(1\)        & \(\color{red}{0}\)                        \\
    \((1+\frac{1}{2}\cdot a)\)                     & \(0\)        & \(\color{red}{\top}\)                     \\
    \hline
    \((0\cdot a)\)                                 & \(0\)        & \(\color{red}{0}\)                        \\
    \((\frac{1}{2}\cdot a)(0\cdot a)\)             & \(0\)        & \(\color{red}{\frac{1}{2}}\)              \\
    \((\frac{1}{2}\cdot a)(\frac{1}{2}\cdot a)\)   & \(1\)        & \(\color{red}{0}\)                        \\
    \((\frac{1}{2}\cdot a)(1\cdot a)\)             & \(0\)        & \(\color{red}{\top}\)                     \\
    \((\frac{1}{2}\cdot a)(1+\frac{1}{2}\cdot a)\) & \(0\)        & \(\color{red}{\top}\)                     \\
    \((1\cdot a)(0\cdot a)\)                       & \(1\)        & \(\color{red}{0}\)                        \\
    \((1\cdot a)(\frac{1}{2}\cdot a)\)             & \(0\)        & \(\color{red}{\frac{1}{2}}\)              \\
    \((1\cdot a)(1\cdot a)\)                       & \(0\)        & \(\color{red}{0}\)                        \\
    \((1\cdot a)(1+\frac{1}{2}\cdot a)\)           & \(0\)        & \(\color{red}{\top}\)                     \\
    \((1+\frac{1}{2}\cdot a)(\Sigma_{K}\cdot a)\)  & \(0\)        & \(\color{red}{\top}\)                     \\
    \hline
  \end{tabular}
  \begin{tabular}{|c|@{\hspace{2pt}}c@{\hspace{2pt}}|c| }
    \hline
    \(\mathcal{T}_{2}\)                                    & \(\epsilon\) & \(\color{red}{\mathit{c}_{\top}^{1}(u)}\) \\
    \hline
    \(\epsilon\)                                           & \(0\)        & \(\color{red}{0}\)                        \\
    \((\frac{1}{2}\cdot a)\)                               & \(0\)        & \(\color{red}{\frac{1}{2}}\)              \\
    \((1\cdot a)\)                                         & \(1\)        & \(\color{red}{0}\)                        \\
    \((1+\frac{1}{2}\cdot a)\)                             & \(0\)        & \(\color{red}{\top}\)                     \\
    \((1\cdot a)(1\cdot a)\)                               & \(0\)        & \(\color{red}{0}\)                        \\
    \((1\cdot a)(1\cdot a)(1\cdot a)\)                     & \(0\)        & \(\color{red}{0}\)                        \\
    \hline
    \((0\cdot a)\)                                         & \(0\)        & \(\color{red}{0}\)                        \\
    \((\frac{1}{2}\cdot a)(0\cdot a)\)                     & \(0\)        & \(\color{red}{\frac{1}{2}}\)              \\
    \((\frac{1}{2}\cdot a)(\frac{1}{2}\cdot a)\)           & \(1\)        & \(\color{red}{0}\)                        \\
    \((\frac{1}{2}\cdot a)(1\cdot a)\)                     & \(0\)        & \(\color{red}{\top}\)                     \\
    \((\frac{1}{2}\cdot a)(1+\frac{1}{2}\cdot a)\)         & \(0\)        & \(\color{red}{\top}\)                     \\
    \((1\cdot a)(0\cdot a)\)                               & \(1\)        & \(\color{red}{0}\)                        \\
    \((1\cdot a)(\frac{1}{2}\cdot a)\)                     & \(0\)        & \(\color{red}{\frac{1}{2}}\)              \\
    \((1\cdot a)(1+\frac{1}{2}\cdot a)\)                   & \(0\)        & \(\color{red}{\top}\)                     \\
    \((1+\frac{1}{2}\cdot a)(\Sigma_{K}\cdot a)\)          & \(0\)        & \(\color{red}{\top}\)                     \\
    \((1\cdot a)(1\cdot a)(0\cdot a)\)                     & \(0\)        & \(\color{red}{0}\)                        \\
    \((1\cdot a)(1\cdot a)(\frac{1}{2}\cdot a)\)           & \(0\)        & \(\color{red}{\frac{1}{2}}\)              \\
    \((1\cdot a)(1\cdot a)(1+\frac{1}{2}\cdot a)\)         & \(0\)        & \(\color{red}{\top}\)                     \\
    \((1\cdot a)(1\cdot a)(1\cdot a)(\Sigma_{K}\cdot a) \) & \(0\)        & \(..\)                                    \\
    \hline
  \end{tabular}
\end{table}
Suppose the unknown IRTA language is
\(L=\{ u\in \mathbb{T}\Sigma^{*} \mid \sigma(u)=1\}\) for
\(\Sigma=\{a\}\) with a known constant \(K=1\).  The learner starts
with \(\mathcal{T}_{0}\) containing only the $\epsilon$ row which is
$0$ since \(\epsilon\notin L\), and the $\epsilon$ column.
\(\mathcal{T}_{0}\) is not closed as witnessed by \((1\cdot a)\),
whose row is $1$, while the row of \(\epsilon\) is $0$.  Additionally,
\((\frac{1}{2}\cdot a)\) and \((1+\frac{1}{2}\cdot a)\) also witness
non-closure because their $\ck_{\top}$ values are non zero.  This
highlights the difference from the untimed case: row words are
distinguished based on their clock values as well.

To obtain a closed table, the learner successively adds the words
\((\frac{1}{2}\cdot a)\), \((1\cdot a)\), \((1+\frac{1}{2}\cdot a)\),
forming \(\mathcal{T}_{1}\).  Every two words in \(\mathcal{T}_{1}\)
are distinguished either by their row or by their $\ck_{\top}$ value.
Thus, \(\mathcal{T}_{1}\) is consistent: if two words have identical
rows and same $\ck_{\top}$ value then their extensions also satisfy
this.  Since \(\mathcal{T}_{1}\) is closed and consistent the learner
conjectures \(\mathcal{A}_{\mathcal{T}_{1}}\)
(Figure~\ref{fig:conjecture} of Appendix~\ref{ap:angluin}), the
$K$-acceptor induced by $\mathcal{T}_{1}$ (formal definition in the
appendix).  The teacher provides a counterexample, assumed to be
\((1\cdot a)(1\cdot a)(1\cdot a)\), which is accepted by
\(\mathcal{A}_{\mathcal{T}_{1}}\) but not in \(L\).  The learner
processes this counterexample and computes \(\mathcal{T}_{2}\), which
is closed but not consistent as shown by \(\epsilon\) and
\((1\cdot a)(1\cdot a)\) and their extensions by \((1\cdot a)\).
Hence, the learner adds a column for \((1\cdot a)\) and computes
\(\mathcal{T}_{3}\).  The process continues similarly The rest of the
run is is detailed in the appendix.

\section{Conclusion}
\label{sec:conclusion}

We have presented a Myhill-Nerode style characterization for timed
languages accepted by timed automata with integer resets.  There are
three main technical ingredients: (1) the notion of $K$-monotonicity
(Definition~\ref{def:monotonicity}) that helps characterize
equivalences on timed words with automata, that we call $K$-acceptors.
This was possible since each word $u$ determines the value $\ck(u)$ of
the clock on reading $u$, in any $K$-acceptor; (2) the definition of
the rescaling function (Section~\ref{sec:nerode-style-equiv}) that
gives a Nerode-like equivalence, leading to a Myhill-Nerode theorem
for IRTA languages, and the canonical equivalence $\approx^{L,K}$; (3)
understanding canonical equivalence $\approx^{L,K}$ through
half-integral words (Section~\ref{sec:canonical-form}), which are, in
some sense, discretized words.  This helps us to build and learn the
canonical form.  We believe these technical ingredients provide
insights into understanding the languages recognized by IRTA.
Typically, active learning algorithms begin by setting up a canonical
form.  We have laid the foundation for IRTAs.  Future work lies in
adapting these foundations for better learning algorithms.

\newpage
\appendix

\section{Proofs of \S~\ref{sec:integer-resets}}
\label{app:strictness}

\PropRPlus*
\begin{proof}
	Given a \(1\)-TA we can always assume that its transitions are of the form given in item 1.
	We call a transition of the form \((q,q',a,x=m,1)\) a \emph{bad} transition.
	Next, we prove by induction on the number of bad transitions that the language of any (deterministic) \(1\)-TA with transitions as in item 1 is also accepted by a (deterministic) \(1\)-TA with no bad transitions.
	Let \(\mathcal{A}=(Q, q_{I}, T, F)\) be a \(1\)-TA with transitions as in item 1 and assume it has \(n+1\) bad transitions.
	Let \(\theta=(p,q,a,x=m,1)\in T\) be a bad transition such that the constant \(m\) is the largest among all the constants associated with bad transitions.

	For a guard \(\phi\) with a constant greater than \(m\), define
	\[
		\mathit{modify}(\phi)=
		\begin{cases}
			x=c-m                    & \quad\text{if \(\phi\) is \(x=c\)}                  \\
			c -m< x \land x < c-m +1 & \quad\text{if \(\phi\) is \(c < x \land x < c +1\)} \\
			K-m<x                    & \quad\text{else (\(\phi\) is \(K<x\))} \enspace .
		\end{cases}
	\]
	Define the \(1\)-TA \(\overline{\mathcal{A}}\) as follows:
	The set of states is given by \(Q\) and a copy \(\overline{Q}\) of \(Q\) and the initial state is \(q_{I}\in Q\).
	The transitions are given by
	\begin{enumerate}
		\item the transitions of \(T\setminus \{\theta\}\),
		\item the transition \((p, \overline{q}, a, x=m , 0)\),
		\item for every non bad transition \(\eta=(s, s', a, \phi, 1)\in T\) the transition \(\overline{\eta}=(\overline{s}, \overline{s}', a, \mathit{modify}(\phi), 1)\),
		\item for every bad transition \(\eta=(s,s',a,x=m,1)\in T\) of guard \(x=m\) (including the transition \(\theta\)) the transition \(\eta=(\overline{s},\overline{s}',a,x=0,0)\),
		\item for every resetting transition \((s, s', a, \phi, 0)\in T\) the transition \((\overline{s}, s', a, \mathit{modify}(\phi), 0)\).
	\end{enumerate}
	The set of final states is \(F\cup \{\overline{f}\in \overline{Q}\mid f \in F\}\).
	Note that if \(\mathcal{A}\) is deterministic then \(\overline{\mathcal{A}}\) is also deterministic.
	All the bad transitions of \(\overline{\mathcal{A}}\) are in \(T\setminus \{\theta\}\), thus \(\overline{\mathcal{A}}\) has \(n\) bad transitions.
	By induction hypothesis it is accepted by a \(1\)-TA without any bad transitions.
	It remains to show that \(L(\mathcal{A})=L(\overline{\mathcal{A}})\).
	The runs of \(\mathcal{A}\) that don't use the transition \(\theta\) correspond to the runs of \(\overline{\mathcal{A}}\) that don't use the transition of item 2.
	\(\overline{\mathcal{A}}\) simulates a run of \(\mathcal{A}\) that uses \(\theta\) as follows:
	Until the first occurrence of \(\theta\) the run is simulated by the transitions in item 1.
	Every occurrence of \(\theta\) in the run is simulated either by the transition in item 2 (in particular this is the case for the first occurrence of \(\theta\)) or by the ``good'' version of \(\theta\) given in item 4.
	After an occurrence of \(\theta\) the automaton \(\overline{\mathcal{A}}\) uses the transitions of items 3 and 4 until a resetting transition occurs.
	Such a resetting transition is simulated with the corresponding transition from item 5 which forces \(\overline{\mathcal{A}}\) back to the states of \(Q\) until a next occurrence of \(\theta\).
\end{proof}

\section{Proofs of \S~\ref{sec:canonical-form}}
\label{app:canonical-form}

\halfIntegral*
\begin{proof}
	If $\ck(u_0)>K$ and $u=(t_1 \cdot a_1) (t_2 \cdot a_2) \dots (t_n \cdot a_n)$ then $w=(0 \cdot a_1) (0 \cdot a_2) \dots (0 \cdot a_n)$ proves the claim.
	Assume now $\ck(u_0)\leq K$.
	We will prove a slightly stronger claim: the word $w$ associated to $u$ is such that $\ck(u_0w)<K+1$.
	Proof proceeds by induction on the length of the timed words.
	Suppose the property is true for all words of length $\le n$.
	Let $u = u_n (t_{n+1} \cdot a_{n+1})$ where $u_n \in \tss$ is a word of length $n$.
	By induction hypothesis, there is a half integral word $w_n$ that witnesses the (stronger) claim for $u_n$.
	We will now construct $w = w_n (t'_{n+1} \cdot a_{n+1})$ that witnesses the property for $u$.

	To satisfy the property, we will define $t'_{n+1}$ such that its fractional part is either $0$ or $\frac{1}{2}$, $\ck(u_0w)<K+1$ and, $\ck(u_0u_n) + t_{n+1} \equiv^K \ck(u_0w_n) + t'_{n+1}$.
	Hence $u_0w$ will be a small half-integral word and the same outgoing transition will be triggered on elapsing $t_{n+1}$ from $u_0u_n$, or while elapsing $t'_{n+1}$ from $u_0w_n$.
	Observe that $\ck(u_0u_n) \equiv^K \ck(u_0w_n)$ and that
	\begin{align}
		\ck(u_0u_n) + t_{n+1}  & = \lfloor \ck(u_0u_n) \rfloor + \lfloor t_{n+1}
		\rfloor + \{ \ck(u_0u_n) \} + \{ t_{n+1} \} \label{eq:9}                 \\
		\ck(u_0w_n) + t'_{n+1} & = \lfloor \ck(u_0w_n) \rfloor + \lfloor
		t'_{n+1} \rfloor + \{ \ck(u_0w_n) \} + \{
		t'_{n+1} \} \label{eq:10}
	\end{align}
	Moreover, when $\ck(u_0u_n) \le K$ we have
	\begin{align}
		\text{ $\lfloor \ck(u_0w_n) \rfloor  = \lfloor \ck(u_0w_n)
    \rfloor$ and $\{ \ck(u_0u_n) \} = 0$ if{}f $\{ \ck(u_0w_n) \} = 0$ } \label{eq:11}
	\end{align}

	\begin{itemize}
		\item   When $K < \ck(u_0u_n)$, we also have $K < \ck(u_0w_n)$.
		      Thus, $\ck(u_0w_n)=K+\frac{1}{2}$ as $u_0w_n$ being small implies its fractional part is either $0$ or \(\frac{1}{2}\).
		      We take $t'_{n+1}=0$.
		\item When $\ck(u_0u_n)=0$ and $t_{n+1}>K$ it suffices to take $t'_{n+1}=K+\frac{1}{2}$.
		\item When $\ck(u_0u_n)=0$ and $t_{n+1}\leq K$ we take $ \lfloor t'_{n+1} \rfloor= \lfloor t_{n+1} \rfloor $ and, $\{ t'_{n+1}\}=\frac{1}{2}$ if $\{ t_{n+1}\}\neq 0$ and $\{ t'_{n+1}\}=0$ otherwise.
		\item When $\ck(u_0u_n)\neq 0$ and $\ck(u_0u_n)\leq K$ we have $\{ \ck(u_0u_n) \} \neq 0$ thus, $\{ \ck(u_0w_n) \} = \frac{1}{2}$ (by \eqref{eq:11} and because $\ck(u_0w_n)$ half-integral).
		      \begin{itemize}
			      \item if $ \ck(u_0u_n)+t_{n+1}>K$ we take $t'_{n+1}=K-\lfloor \ck(u_0w_n) \rfloor$.
			      \item if $ \ck(u_0u_n)+t_{n+1}\leq K$ then
			            \begin{align*}
				            t'_{n+1} = \begin{cases}
					                       \lfloor t_{n+1} \rfloor               & \text{ if } 0< \{ \ck(u_0u_n) \} +
					                       \{t_{n+1}\} <1                                                             \\
					                       \lfloor t_{n+1} \rfloor + \frac{1}{2} & \text{ if  }
					                       \{\ck(u_0u_n)\} + \{t_{n+1}\} = 1                                          \\
				                       \end{cases}
			            \end{align*}
		      \end{itemize}
	\end{itemize}
\end{proof}

\canonicalHalfIntegral*
\begin{proof}
	If $\approx^\Bb$ and $\approx^{L,K}$ coincide then in particular they coincide on half-integral words.
	For the reverse direction, since $\approx^\Bb$ refines $\approx^{L,K}$ it suffices to show that $u \approx^{L, K} v$ implies $u \approx^\Bb v$ for $u,v \in \mathbb{T}\Sigma^{*}$.
	Let $u'$ be the half-integral word of Lemma~\ref{lem:prelim} for $u_0=\epsilon$ and $u$.
	Define similarly $v'$.
	We have $u', v' \in (\Sigma_K)^*$.
	The $K$-acceptor $\Bb$ reaches the same state on reading $u$ or $u'$, thus $u \approx^\Bb u'$.
	Similarly $v \approx^\Bb v'$.
	Since $\approx^\Bb$ refines $\approx^{L,K}$ we have $u \approx^{L, K} u'$ and $v \approx^{L, K} v'$ and by transitivity, $u' \approx^{L, K} v'$.
	Hence, $u' \approx^\Bb v'$ because $u',v'\in(\Sigma_K)^*$ and finally by transitivity $u \approx^\Bb v$.
\end{proof}

\section{Details for Angluin's Learning}
\label{ap:angluin}

\forLearning*
\begin{proof}
	Since $u \approx^{L, K} v$ and the fractions $\{\ck(u)\},\{\ck(v)\}$ are in $\{0, \frac{1}{2}\}$, \(K\)-monotonicity \textcolor{lipicsGray}{\sffamily\bfseries\upshape (a)} implies $\ck_{\top}(u)=\ck_{\top}(v)$.
	Thus, \(\tau_{u \to v}\) is equal to the identity function.
  Thus, $uz \in L$ if{}f $vz \in L$.
	For the converse direction, using that \(\tau_{u \to v}\) is the identity function, we need to show that $u^{-1}L=v^{-1}L$.
	Let $w\in u^{-1}L$.
  By Lemma~\ref{lem:prelim} (with $u_0=u$) there is a word $w'$ such that $uw' \in (\Sigma_K)^*$ and $uw\in L$ if{}f $uw'\in L$.
	Since the residuals $u^{-1}L$ and $v^{-1}L$ contain the same words in $(\Sigma_K)^*$ and $w' \in (\Sigma_K)^*$, from $uw'\in L$ we find $vw'\in L$.
	Since every $K$-acceptor reaches the same state on reading $uw'$ or $uw$ and since $\ck_{\top}(u)=\ck_{\top}(v)$, every $K$-acceptor reaches the same state on reading $vw'$ or $vw$.
	Thus, $vw'\in L$ implies $vw\in L$.
	Thus, $w\in v^{-1}L$.
\end{proof}

In this section we adapt the learning framework of Angluin~\cite{angluinLearningRegularSets1987} to learn \(1\)-IRDTA.
\subparagraph{Observation tables.}
An \emph{Observation Table} $\mathcal{T}$ is a matrix of $0$'s and $1$'s with rows and columns labeled by half-integral words in \((\Sigma_{K})^{*}\) given by a prefix-closed subset $U_{1}$ for the rows and a suffix-closed subset \(E\) for the columns.
Additionally, the set $U_2$ contains the extensions of $U_1$ i.e., the words obtained by concatenating a letter from $\Sigma_K$ to the words in $U_1$.
We write $\mathcal{T}(u,e)$ for the entry in row $u$ and column $e$.

For $u\in U_1\cup U_2$ let $\mathit{R}(u)=(\mathit{row}(u), \ck_{\top}(u))$ where $\mathit{row}(u)$ is the row of $u$.
A table is closed when for every $w\in U_2$ there is $u\in U_1$ such that $\mathit{R}(w)=\mathit{R}(u)$.
It is consistent when for every $u,w\in U_1$ such that $\mathit{R}(w)=\mathit{R}(u)$ we have $\mathit{R}(w(t\cdot a))=\mathit{R}(u(t\cdot a))$ for every $a\in \Sigma$ and $t\in \Sigma_K$.

\subparagraph{Acceptor of an observation table.}
A closed and consistent observation table $\mathcal{T}$ induces a $K$-acceptor \(\mathcal{A}_{\mathcal{T}}\): The states are given by $\mathit{R}(u)$ for $u\in U_1$.
The initial state is $\mathit{R}(\epsilon)$ and a state $\mathit{R}(u)$ is final if $\mathcal{T}(u,\epsilon)=1$.
We have transitions \((\mathit{R}(u), \mathit{R}(u(t\cdot a)),a,\phi_{[\mathit{c}^{K}(u)+t]_{\equiv^{K}}}, s)\) for every $a\in \Sigma$ and $t\in \Sigma_{K}$, where $s=0$ if{}f $\mathit{c}^{K}(u)+t\in \{0, \dots, K\}$.

A $K$-acceptor is \emph{consistent with $\mathcal{T}$} if for every \(w\in U_1\cup U_2\) and \(e\in E\) such that \((q_{I},0) \rightsquigarrow^{we} (q, \mathit{c}^{K}(we))\) we have \(q\) is a final state if{}f \(\mathcal{T}(w,e)=1\).
\begin{restatable}{theorem}{ThmAngluin}
	If \(\mathcal{T}\) is closed and consistent then \(\mathcal{A}_{\mathcal{T}}\) is consistent with \(\mathcal{T}\).
	No other acceptor consistent with \(\mathcal{T}\) has equally many or less states than \(\mathcal{A}_{\mathcal{T}}\).
\end{restatable}
\begin{proof}
	Consequence of Lemma~\ref{lemma:angluin2} and Lemma~\ref{lemma:angluin3}.
\end{proof}
\begin{lemma}
	\label{lemma:angluin1}
	If \(\mathcal{T}\) is closed and consistent then for every \(w\in U_1\cup U_2\) we have \((\mathit{R}(\epsilon),0) \rightsquigarrow^{w} (\mathit{R}(w), \mathit{c}^{K}(w))\) in \(\mathcal{A}_{\mathcal{T}}\).
\end{lemma}
\begin{proof}
	The proof goes by induction on the length of \(w\in U_1\cup U_2\) by noticing that \((\mathit{R}(\epsilon),0) \rightsquigarrow^{w} (\mathit{R}(w), \mathit{c}^{K}(w))\) and \(w(t \cdot a)\in U_1\cup U_2 \) then there is a transition \((\mathit{R}(w), \mathit{R}(w(t\cdot a)),a,\phi_{[\mathit{c}^{K}(u)+t]_{\equiv^{K}}}, s)\).
\end{proof}
\begin{lemma}
	\label{lemma:angluin2}
	If \(\mathcal{T}\) is closed and consistent then \(\mathcal{A}_{\mathcal{T}}\) is consistent with \(\mathcal{T}\).
\end{lemma}
\begin{proof}
	By induction on the length of \(e\) s.t.
	\((\mathit{R}(\epsilon),0) \rightsquigarrow^{we} (q, \mathit{c}^{K}(we))\).
	If \(e\) is the empty word then by Lemma~\ref{lemma:angluin1} \(q=\mathit{R}(w)\).
  If \(w\in U_1\) then by definition \(q\) is a final state if{}f \(\mathcal{T}(w,\epsilon)=1\).
	If \(w\in U_2\) then since \(\mathcal{T}\) is closed \(q=\mathit{R}(u)\) for some \(u\in U_1\), and \(\mathit{R}(u)\) is a final state if{}f \(\mathcal{T}(u,\epsilon)=1\) if{}f \(\mathcal{T}(w,\epsilon)=1\) since \(\mathit{R}(u)=\mathit{R}(w)\).
	Suppose the result holds for all elements in \(E\) no longer than \(k\), and let \(e\in E\) of length \(k+1\).
	Since \(E\) is suffix-closed, \(e=(t \cdot a)e_{1}\) for some \(e_{1}\in E\) of length \(k\).
	Let \(w\in U_1\cup U_2 \) such that \((\mathit{R}(\epsilon),0) \rightsquigarrow^{we} (q, \mathit{c}^{K}(we))\).
	By Lemma~\ref{lemma:angluin1} and since \(\mathcal{T}\) is closed we have \((\mathit{R}(\epsilon),0) \rightsquigarrow^{w} (\mathit{R}(w), \mathit{c}^{K}(w))=(\mathit{R}(u), \mathit{c}^{K}(u))\) for some \(u\in U_1\).
	By definition of \(U_2 \), \(u(t \cdot a)\in U_2\).
	Hence, using Lemma~\ref{lemma:angluin1}, we deduce from \((\mathit{R}(\epsilon),0) \rightsquigarrow^{we} (q, \mathit{c}^{K}(we))\) the run \((\mathit{R}(\epsilon),0) \rightsquigarrow^{u(t \cdot a)} (\mathit{R}(u(t \cdot a)), \mathit{c}^{K}(u(t \cdot a))) \rightsquigarrow^{e_{1}} (q_{}, \mathit{c}^{K}(u(t \cdot a)e_{1}))\).
	By induction hypothesis \(q\) is a final state if{}f \(\mathcal{T}(u(t \cdot a), e_{1})=1\) if{}f \(\mathcal{T}(w( t\cdot a),e_{1})=1\) since \(\mathit{R}(u)=\mathit{R}(w)\) and \(\mathcal{T}\) is consistent.
\end{proof}
\begin{lemma}
	\label{lemma:angluin3}
	Assume \(\mathcal{T}\) is closed and consistent.
	Every $K$-acceptor consistent with \(\mathcal{T}\) that has equally many or less states than \(\mathcal{A}_{\mathcal{T}}\) is isomorphic to \(\mathcal{A}_{\mathcal{T}}\).
\end{lemma}
\begin{proof}
	Let \(\mathcal{A}=(Q,q_{I},T,F)\) be a $K$-acceptor consistent with \(\mathcal{T}\) that has equally many or less states than \(\mathcal{A}_{\mathcal{T}}\).
	For each \(q\in Q\), let \(\mathit{row}(q) \colon E \rightarrow \{0,1\}\) such that \(\mathit{row}(q)(e)=1\) if{}f \((q, r) \rightsquigarrow^{e} (p, x)\) holds for some \(p\in F\) and \(r\in \mathit{region}(q)\) with \(\mathit{frac}(r)\in\{0, \frac{1}{2}\}\).
	Since \(\mathcal{A}\) is complete and deterministic we can define the function \(f:U_1\cup U_2 \to Q\) mapping \(w\in U_1\cup U_2 \) onto the unique state \(q\in Q\) such that \((q_{I},0) \rightsquigarrow^{w} (q, \mathit{c}^{K}(w))\).
	Since \(\mathcal{A}\) is consistent with \(\mathcal{T}\) for every \(w\in U_1\cup U_2 \) we have \(\mathit{R}(w)=(\mathit{row}(f(w)), \ck_{\top}(w))\).

	Next we show that we have the following property \((\star)\): for every \(w,w'\in U_1\cup U_2 \), \(f(w)=f(w')\iff \mathit{R}(w)=\mathit{R}(w')\).
	Assume \(f(w)=f(w')=q\).
	Clearly, \(\mathit{c}^{K}_{\top}(w)= \mathit{c}^{K}_{\top}(w')\).
	From \(\mathit{R}(w)=(\mathit{row}(f(w)), \ck_{\top}(w))\) and \(f(w)=f(w')\) we find that \(\mathit{R}(w)=\mathit{R}(w')\).
	Therefore, two words \(w,w'\in U_1\cup U_2\) such that \(\mathit{R}(w) \neq \mathit{R}(w')\) cannot lead to the same state of \(\mathcal{A}\).
	Thus, \(\mathcal{A}\) has at least as many states as \(\mathcal{A}_{\mathcal{T}}\).
	Thus, it has exactly the same number of states as \(\mathcal{A}_{\mathcal{T}}\).
	Hence, \(\mathit{R}(w)=\mathit{R}(w')\implies f(w)=f(w')\).

	Let $h$ be the function which maps a state $\mathit{R}(w)$ of \(\mathcal{A}_{\mathcal{T}}\) to \(h(\mathit{R}(w))=f(w)\in Q\).
	By the property \((\star)\), \(h\) is a bijection.
	Moreover, \(h(\mathit{R}(\epsilon))=q_{I}\) and $h$ maps the final states \(\mathcal{A}_{\mathcal{T}}\) to the final states of \(\mathcal{A}\).
	For \(\mathcal{A}\) and \(\mathcal{A}_{\mathcal{T}}\) to be isomorphic we show that $h$ also preserves the transitions: \((\mathit{R}(u), \mathit{R}(v),a,\phi_{[\mathit{c}^{K}(u)+t]_{\equiv^{K}}}, s)\) is a transition of \(\mathcal{A}_{\mathcal{T}}\) if{}f \((h(\mathit{R}(u)), h(\mathit{R}(v)),a,\phi_{[\mathit{c}^{K}(u)+t]_{\equiv^{K}}}, s)\) is a transition of \(\mathcal{A}\).
	For the direct direction: \((\mathit{R}(u), \mathit{R}(v),a,\phi_{[\mathit{c}^{K}(u)+t]_{\equiv^{K}}}, s)\) is a transition of \(\mathcal{A}_{\mathcal{T}}\) if{}f \(\mathit{R}(u(t\cdot a))=\mathit{R}({v})\) for some \(t\in (\Sigma_{K})\).
	Since \(\mathcal{A}\) is complete we always have \((f(u), f(u(t\cdot a)),a,\phi_{[\mathit{c}^{K}(u)+t]_{\equiv^{K}}}, s)\in T\), thus \((h(\mathit{R}(u)), h(\mathit{R}(v)),a,\phi_{[\mathit{c}^{K}(u)+t]_{\equiv^{K}}}, s)\in T\).
	For the reverse direction: since \(\mathcal{A}\) is deterministic \((f(u), f(v),a,\phi_{[\mathit{c}^{K}(u)+t]_{\equiv^{K}}}, s)\in T\) implies \(f(v)=f(u(t\cdot a))\).
	We thus conclude by using the property \((\star)\).
\end{proof}

\subsection{Example Run}
Suppose the unknown IRTA language is \(L=\{ u\in \mathbb{T}\Sigma^{*} \mid \sigma(u)=1\}\) for \(\Sigma=\{a\}\).
The constant \(K=1\) is known by the learner.
The table \(\mathcal{T}_{0}\) is consistent but not closed.
In order to get a closed table the learner successively adds the words \((\frac{1}{2}\cdot a)\), \((1\cdot a)\), \((1+\frac{1}{2}\cdot a)\) to finally get the table \(\mathcal{T}_{1}\).
Since this table is closed and consistent the learner makes the conjecture \(\mathcal{A}_{\mathcal{T}_{1}}\) given in Figure~\ref{fig:conjecture}.
The teacher selects a counterexample.
We assume this counterexample to be \((1\cdot a)(1\cdot a)(1\cdot a)\) which is accepted by \(\mathcal{A}_{\mathcal{T}_{1}}\) but not in \(L\).
The learner process this counterexample and computes \(\mathcal{T}_{2}\).
The table \(\mathcal{T}_{2}\) is closed but not consistent since \(\mathit{R}({\epsilon})=\mathit{R}((1\cdot a)(1\cdot a))\) but \(\mathit{R}((1\cdot a))\neq \mathit{R}((1\cdot a)(1\cdot a)(1\cdot a))\).
Hence the learner adds a column for \((1\cdot a)\) and computes \(\mathcal{T}_{3}\) that is closed and consistent.
Next, the teacher returns the counterexample \( (1\cdot a)(\frac{1}{2}\cdot a)(\frac{1}{2}\cdot a)\) to the conjecture \(\mathcal{A}_{\mathcal{T}_{3}}\) given in Figure~\ref{fig:conjecture5}.
Finally, \(\mathcal{T}_{4}\) is obtained by adding \((1\cdot a)(\frac{1}{2}\cdot a)(\frac{1}{2}\cdot a)\) and its prefix \((1\cdot a)(\frac{1}{2}\cdot a)\) to \(U_1\).
However, \(\mathcal{T}_{4}\) is not consistent since \(\mathit{R}((\frac{1}{2}\cdot a))=\mathit{R}((1\cdot a)(\frac{1}{2}\cdot a))\) but \(\mathit{R}((\frac{1}{2}\cdot a)(\frac{1}{2}\cdot a)) \neq \mathit{R}((1\cdot a)(\frac{1}{2}\cdot a)(\frac{1}{2}\cdot a))\).
Hence the learner adds a column for \((\frac{1}{2}\cdot a)\) and computes \(\mathcal{T}_{5}\) that is closed and consistent.
Finally, the learner conjectures \(\mathcal{A}_{\mathcal{T}_{5}}\), which corresponds to the correct \(1\)-IRTA in Figure~\ref{fig:ta_L_R} and he is done.
\begin{tabular}{|c|@{\hspace{2pt}}c@{\hspace{2pt}}|c|}
	\hline
	\(\mathcal{T}_{0}\)        & \(\epsilon\) & \(\color{red}{\mathit{c}_{\top}^{1}(u)}\) \\
	\hline
	\(\epsilon\)               & \(0\)        & \(\color{red}{0}\)                        \\
	\hline
	\((0\cdot a)\)             & \(0\)        & \(\color{red}{0}\)                        \\
	\((\frac{1}{2}\cdot a)\)   & \(0\)        & \(\color{red}{\frac{1}{2}}\)              \\
	\((1\cdot a)\)             & \(1\)        & \(\color{red}{0}\)                        \\
	\((1+\frac{1}{2}\cdot a)\) & \(0\)        & \(\color{red}{\top}\)                     \\
	\hline
\end{tabular}
\begin{tabular}{|c|@{\hspace{2pt}}c@{\hspace{2pt}}|c|}
	\hline
	\(\mathcal{T}_{1}\)                            & \(\epsilon\) & \(\color{red}{\mathit{c}_{\top}^{1}(u)}\) \\
	\hline
	\(\epsilon\)                                   & \(0\)        & \(\color{red}{0}\)                        \\
	\((\frac{1}{2}\cdot a)\)                       & \(0\)        & \(\color{red}{\frac{1}{2}}\)              \\
	\((1\cdot a)\)                                 & \(1\)        & \(\color{red}{0}\)                        \\
	\((1+\frac{1}{2}\cdot a)\)                     & \(0\)        & \(\color{red}{\top}\)                     \\
	\hline
	\((0\cdot a)\)                                 & \(0\)        & \(\color{red}{0}\)                        \\
	\((\frac{1}{2}\cdot a)(0\cdot a)\)             & \(0\)        & \(\color{red}{\frac{1}{2}}\)              \\
	\((\frac{1}{2}\cdot a)(\frac{1}{2}\cdot a)\)   & \(1\)        & \(\color{red}{0}\)                        \\
	\((\frac{1}{2}\cdot a)(1\cdot a)\)             & \(0\)        & \(\color{red}{\top}\)                     \\
	\((\frac{1}{2}\cdot a)(1+\frac{1}{2}\cdot a)\) & \(0\)        & \(\color{red}{\top}\)                     \\
	\((1\cdot a)(0\cdot a)\)                       & \(1\)        & \(\color{red}{0}\)                        \\
	\((1\cdot a)(\frac{1}{2}\cdot a)\)             & \(0\)        & \(\color{red}{\frac{1}{2}}\)              \\
	\((1\cdot a)(1\cdot a)\)                       & \(0\)        & \(\color{red}{0}\)                        \\
	\((1\cdot a)(1+\frac{1}{2}\cdot a)\)           & \(0\)        & \(\color{red}{\top}\)                     \\
	\((1+\frac{1}{2}\cdot a)(\Sigma_{K}\cdot a)\)  & \(0\)        & \(\color{red}{\top}\)                     \\
	\hline
\end{tabular}
\begin{tabular}{|c|@{\hspace{2pt}}c@{\hspace{2pt}}|c| }
	\hline
	\(\mathcal{T}_{2}\)                                    & \(\epsilon\) & \(\color{red}{\mathit{c}_{\top}^{1}(u)}\) \\
	\hline
	\(\epsilon\)                                           & \(0\)        & \(\color{red}{0}\)                        \\
	\((\frac{1}{2}\cdot a)\)                               & \(0\)        & \(\color{red}{\frac{1}{2}}\)              \\
	\((1\cdot a)\)                                         & \(1\)        & \(\color{red}{0}\)                        \\
	\((1+\frac{1}{2}\cdot a)\)                             & \(0\)        & \(\color{red}{\top}\)                     \\
	\((1\cdot a)(1\cdot a)\)                               & \(0\)        & \(\color{red}{0}\)                        \\
	\((1\cdot a)(1\cdot a)(1\cdot a)\)                     & \(0\)        & \(\color{red}{0}\)                        \\
	\hline
	\((0\cdot a)\)                                         & \(0\)        & \(\color{red}{0}\)                        \\
	\((\frac{1}{2}\cdot a)(0\cdot a)\)                     & \(0\)        & \(\color{red}{\frac{1}{2}}\)              \\
	\((\frac{1}{2}\cdot a)(\frac{1}{2}\cdot a)\)           & \(1\)        & \(\color{red}{0}\)                        \\
	\((\frac{1}{2}\cdot a)(1\cdot a)\)                     & \(0\)        & \(\color{red}{\top}\)                     \\
	\((\frac{1}{2}\cdot a)(1+\frac{1}{2}\cdot a)\)         & \(0\)        & \(\color{red}{\top}\)                     \\
	\((1\cdot a)(0\cdot a)\)                               & \(1\)        & \(\color{red}{0}\)                        \\
	\((1\cdot a)(\frac{1}{2}\cdot a)\)                     & \(0\)        & \(\color{red}{\frac{1}{2}}\)              \\
	\((1\cdot a)(1+\frac{1}{2}\cdot a)\)                   & \(0\)        & \(\color{red}{\top}\)                     \\
	\((1+\frac{1}{2}\cdot a)(\Sigma_{K}\cdot a)\)          & \(0\)        & \(\color{red}{\top}\)                     \\
	\((1\cdot a)(1\cdot a)(0\cdot a)\)                     & \(0\)        & \(\color{red}{0}\)                        \\
	\((1\cdot a)(1\cdot a)(\frac{1}{2}\cdot a)\)           & \(0\)        & \(\color{red}{\frac{1}{2}}\)              \\
	\((1\cdot a)(1\cdot a)(1+\frac{1}{2}\cdot a)\)         & \(0\)        & \(\color{red}{\top}\)                     \\
	\((1\cdot a)(1\cdot a)(1\cdot a)(\Sigma_{K}\cdot a) \) & \(0\)        & \(..\)                                    \\
	\hline
\end{tabular}
\begin{tabular}{|c|@{\hspace{2pt}}c@{\hspace{2pt}}|@{}c@{}|c|}
	\hline
	\(\mathcal{T}_{3}\)                                    & \(\epsilon\) & \((1{\cdot}a)\) & \(\color{red}{\mathit{c}_{\top}^{1}(u)}\) \\
	\hline
	\(\epsilon\)                                           & \(0\)        & \(1\)           & \(\color{red}{0}\)                        \\
	\((\frac{1}{2}\cdot a)\)                               & \(0\)        & \(0\)           & \(\color{red}{\frac{1}{2}}\)              \\
	\((1\cdot a)\)                                         & \(1\)        & \(0\)           & \(\color{red}{0}\)                        \\
	\((1+\frac{1}{2}\cdot a)\)                             & \(0\)        & \(0\)           & \(\color{red}{\top}\)                     \\
	\((1\cdot a)(1\cdot a)\)                               & \(0\)        & \(0\)           & \(\color{red}{0}\)                        \\
	\((1\cdot a)(1\cdot a)(1\cdot a)\)                     & \(0\)        & \(0\)           & \(\color{red}{0}\)                        \\
	\hline
	\((0\cdot a)\)                                         & \(0\)        & \(1\)           & \(\color{red}{0}\)                        \\
	\((\frac{1}{2}\cdot a)(0\cdot a)\)                     & \(0\)        & \(0\)           & \(\color{red}{\frac{1}{2}}\)              \\
	\((\frac{1}{2}\cdot a)(\frac{1}{2}\cdot a)\)           & \(1\)        & \(0\)           & \(\color{red}{0}\)                        \\
	\((\frac{1}{2}\cdot a)(1\cdot a)\)                     & \(0\)        & \(0\)           & \(\color{red}{\top}\)                     \\
	\((\frac{1}{2}\cdot a)(1+\frac{1}{2}\cdot a)\)         & \(0\)        & \(0\)           & \(\color{red}{\top}\)                     \\

	\((1\cdot a)(0\cdot a)\)                               & \(1\)        & \(0\)           & \(\color{red}{0}\)                        \\
	\((1\cdot a)(\frac{1}{2}\cdot a)\)                     & \(0\)        & \(0\)           & \(\color{red}{\frac{1}{2}}\)              \\
	\((1\cdot a)(1+\frac{1}{2}\cdot a)\)                   & \(0\)        & \(0\)           & \(\color{red}{\top}\)                     \\
	\((1+\frac{1}{2}\cdot a)(\Sigma_{K}\cdot a)\)          & \(0\)        & \(0\)           & \(\color{red}{\top}\)                     \\
	\((1\cdot a)(1\cdot a)(0\cdot a)\)                     & \(0\)        & \(0\)           & \(\color{red}{0}\)                        \\
	\((1\cdot a)(1\cdot a)(\frac{1}{2}\cdot a)\)           & \(0\)        & \(0\)           & \(\color{red}{\frac{1}{2}}\)              \\
	\((1\cdot a)(1\cdot a)(1+\frac{1}{2}\cdot a)\)         & \(0\)        & \(0\)           & \(\color{red}{\top}\)                     \\
	\((1\cdot a)(1\cdot a)(1\cdot a)(\Sigma_{K}\cdot a) \) & \(0\)        & \(0\)           & \(..\)                                    \\
	\hline
\end{tabular}

\begin{tabular}{|c|c|@{}c@{}|c|}
	\hline
	\(\mathcal{T}_{4}\)                                                       & \(\epsilon\) & \((1{\cdot}a)\) & \(\color{red}{\mathit{c}_{\top}^{1}(u)}\) \\
	\hline
	\(\epsilon\)                                                              & \(0\)        & \(1\)           & \(\color{red}{0}\)                        \\
	\((\frac{1}{2}\cdot a)\)                                                  & \(0\)        & \(0\)           & \(\color{red}{\frac{1}{2}}\)              \\
	\((1\cdot a)\)                                                            & \(1\)        & \(0\)           & \(\color{red}{0}\)                        \\
	\((1+\frac{1}{2}\cdot a)\)                                                & \(0\)        & \(0\)           & \(\color{red}{\top}\)                     \\
	\((1\cdot a)(1\cdot a)\)                                                  & \(0\)        & \(0\)           & \(\color{red}{0}\)                        \\
	\((1\cdot a)(1\cdot a)(1\cdot a)\)                                        & \(0\)        & \(0\)           & \(\color{red}{0}\)                        \\
	\((1\cdot a)(\frac{1}{2}\cdot a)\)                                        & \(0\)        & \(0\)           & \(\color{red}{\frac{1}{2}}\)              \\
	\((1\cdot a)(\frac{1}{2}\cdot a)(\frac{1}{2}\cdot a)\)                    & \(0\)        & \(0\)           & \(\color{red}{0}\)                        \\

	\hline
	\((0\cdot a)\)                                                            & \(0\)        & \(0\)           & \(\color{red}{0}\)                        \\
	\((\frac{1}{2}\cdot a)(0\cdot a)\)                                        & \(0\)        & \(1\)           & \(\color{red}{\frac{1}{2}}\)              \\
	\((\frac{1}{2}\cdot a)(\frac{1}{2}\cdot a)\)                              & \(1\)        & \(0\)           & \(\color{red}{0}\)                        \\
	\((\frac{1}{2}\cdot a)(1\cdot a)\)                                        & \(0\)        & \(0\)           & \(\color{red}{\top}\)                     \\
	\((1\cdot a)(0\cdot a)\)                                                  & \(1\)        & \(0\)           & \(\color{red}{0}\)                        \\
	\((1\cdot a)(1+\frac{1}{2}\cdot a)\)                                      & \(0\)        & \(0\)           & \(\color{red}{\top}\)                     \\
	\((1+\frac{1}{2}\cdot a)(0\cdot a)\)                                      & \(0\)        & \(0\)           & \(\color{red}{\top}\)                     \\
	\((1\cdot a)(1\cdot a)(0\cdot a)\)                                        & \(0\)        & \(0\)           & \(\color{red}{0}\)                        \\
	\((1\cdot a)(1\cdot a)(\frac{1}{2}\cdot a)\)                              & \(0\)        & \(0\)           & \(\color{red}{\frac{1}{2}}\)              \\
	\((1\cdot a)(1\cdot a)(1+\frac{1}{2}\cdot a)\)                            & \(0\)        & \(0\)           & \(\color{red}{\top}\)                     \\
	\((1\cdot a)(1\cdot a)(1\cdot a)(\Sigma_{K}\cdot a)\)                     & \(0\)        & \(0\)           & \(..\)                                    \\
	\((1\cdot a)(\frac{1}{2}\cdot a)(\Sigma_{K}\cdot a)\)                     & \(0\)        & \(0\)           & \(..\)                                    \\
	\((1\cdot a)(\frac{1}{2}\cdot a)(\frac{1}{2}\cdot a)(\Sigma_{K}\cdot a)\) & \(0\)        & \(0\)           & \(..\)                                    \\
	\hline
\end{tabular}
\begin{tabular}{|c|@{\hspace{2pt}}c@{\hspace{2pt}}|@{}c@{}|@{}c@{}|c|}
	\hline
	\(\mathcal{T}_{5}\)                                                       & \(\epsilon\) & \((1{\cdot}a)\) & \((\frac{1}{2}\cdot a)\) & \(\color{red}{\mathit{c}_{\top}^{1}(u)}\) \\
	\hline
	\(\epsilon\)                                                              & \(0\)        & \(1\)           & \(0\)                    & \(\color{red}{0}\)                        \\
	\((\frac{1}{2}\cdot a)\)                                                  & \(0\)        & \(0\)           & \(1\)                    & \(\color{red}{\frac{1}{2}}\)              \\
	\((1\cdot a)\)                                                            & \(1\)        & \(0\)           & \(0\)                    & \(\color{red}{0}\)                        \\
	\((1+\frac{1}{2}\cdot a)\)                                                & \(0\)        & \(0\)           & \(0\)                    & \(\color{red}{\top}\)                     \\
	\((1\cdot a)(1\cdot a)\)                                                  & \(0\)        & \(0\)           & \(0\)                    & \(\color{red}{0}\)                        \\
	\((1\cdot a)(1\cdot a)(1\cdot a)\)                                        & \(0\)        & \(0\)           & \(0\)                    & \(\color{red}{0}\)                        \\
	\((1\cdot a)(\frac{1}{2}\cdot a)\)                                        & \(0\)        & \(0\)           & \(0\)                    & \(\color{red}{\frac{1}{2}}\)              \\
	\((1\cdot a)(\frac{1}{2}\cdot a)(\frac{1}{2}\cdot a)\)                    & \(0\)        & \(0\)           & \(0\)                    & \(\color{red}{0}\)                        \\

	\hline
	\((0\cdot a)\)                                                            & \(0\)        & \(0\)           & \(0\)                    & \(\color{red}{0}\)                        \\
	\((\frac{1}{2}\cdot a)(0\cdot a)\)                                        & \(0\)        & \(1\)           & \(1\)                    & \(\color{red}{\frac{1}{2}}\)              \\
	\((\frac{1}{2}\cdot a)(\frac{1}{2}\cdot a)\)                              & \(1\)        & \(0\)           & \(0\)                    & \(\color{red}{0}\)                        \\
	\((\frac{1}{2}\cdot a)(1\cdot a)\)                                        & \(0\)        & \(0\)           & \(0\)                    & \(\color{red}{\top}\)                     \\
	\((1\cdot a)(0\cdot a)\)                                                  & \(1\)        & \(0\)           & \(0\)                    & \(\color{red}{0}\)                        \\
	\((1\cdot a)(1+\frac{1}{2}\cdot a)\)                                      & \(0\)        & \(0\)           & \(0\)                    & \(\color{red}{\top}\)                     \\
	\((1+\frac{1}{2}\cdot a)(0\cdot a)\)                                      & \(0\)        & \(0\)           & \(0\)                    & \(\color{red}{\top}\)                     \\
	\((1\cdot a)(1\cdot a)(0\cdot a)\)                                        & \(0\)        & \(0\)           & \(0\)                    & \(\color{red}{0}\)                        \\
	\((1\cdot a)(1\cdot a)(\frac{1}{2}\cdot a)\)                              & \(0\)        & \(0\)           & \(0\)                    & \(\color{red}{\frac{1}{2}}\)              \\
	\((1\cdot a)(1\cdot a)(1+\frac{1}{2}\cdot a)\)                            & \(0\)        & \(0\)           & \(0\)                    & \(\color{red}{\top}\)                     \\
	\((1\cdot a)(1\cdot a)(1\cdot a)(\Sigma_{K}\cdot a)\)                     & \(0\)        & \(0\)           & \(0\)                    & \(..\)                                    \\
	\((1\cdot a)(\frac{1}{2}\cdot a)(\Sigma_{K}\cdot a)\)                     & \(0\)        & \(0\)           & \(0\)                    & \(..\)                                    \\
	\((1\cdot a)(\frac{1}{2}\cdot a)(\frac{1}{2}\cdot a)(\Sigma_{K}\cdot a)\) & \(0\)        & \(0\)           & \(0\)                    & \(..\)                                    \\

	\hline
\end{tabular}
\begin{figure}[p]
	\centering
	\begin{tikzpicture}[shorten >=1pt,node distance=6cm,on grid,auto,initial text={}]
		\node[state,initial] (epsilon) at (0,0) {\(\epsilon\)};
		\node[state] (s) at (5,-2) {\((\frac{1}{2}\cdot a)\)};
		\node[state, accepting] (q) at (10,1) {\((1\cdot a)\)};
		\node[state] (2) at (5,0) {\((1+\frac{1}{2} \cdot a)\)};
		\path[->]
			(epsilon) edge [loop above]             node                              {\(a,x=0,0\)} ()
			(epsilon) edge node [below, sloped]              {\(a,0<x<1,1\)} (s)
			(epsilon) edge [bend left]              node [below, sloped, near start]  {\(a,x=1,0\)} (q)
			(epsilon) edge node                              {\(a,1<x,1\)} (2)
			(s)       edge [loop left]              node [below, left]                {\(a,0<x<1,1\)} ()
			(s)       edge [bend left=10]           node [below, sloped]              {\(a,x=1,0\)} (q)
			(s)       edge node                              {\(a,1<x,1\)} (2)
			(q)       edge [loop above]             node                              {\(a,x=0,0\)} ()
			(q)       edge [bend left=10]           node [below, sloped]              {\(a,0<x<1,1\)} (s)
			(q)       edge [above]                  node [sloped]                     {\(a,1<x,1\)} (2)
			(q)       edge [bend right=45]          node [below]                      {\(a,x=1,0\)} (epsilon)
			(2)       edge [loop above]             node [right, near end]            {\(a,1<x,1\)} ();
	\end{tikzpicture}
	\caption{Conjecture \(\mathcal{A}_{\mathcal{T}_{1}}\).}
	\label{fig:conjecture}
\end{figure}
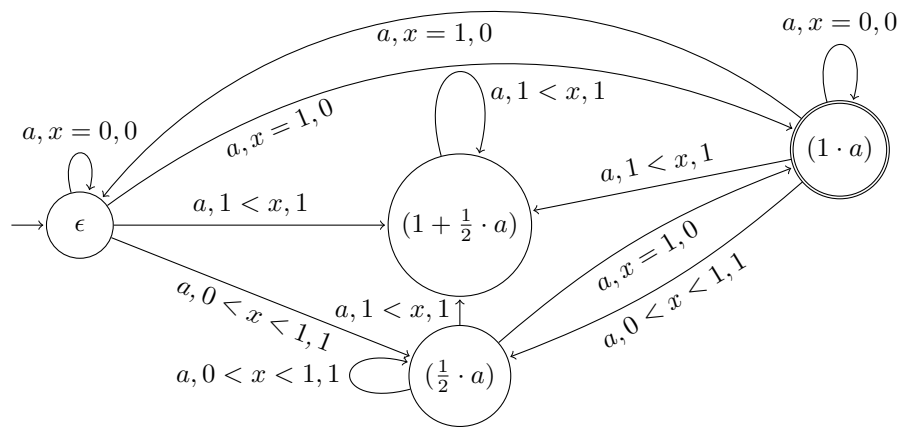
\begin{figure}[p]
	\centering
	\begin{tikzpicture}[shorten >=1pt,node distance=6cm,on grid,auto,initial text={}]
		\node[state,initial]    (epsilon)   at (0,0)    {\(\epsilon\)};
		\node[state]            (s)         at (5,-2)   {\((\frac{1}{2}\cdot a)\)};
		\node[state,accepting]  (q)         at (10,1.5) {\((1\cdot a)\)};
		\node[state]            (2)         at (5,0)    {\((1+\frac{1}{2} \cdot a)\)};
		\node[state]            (r)         at (10,-2)  {\((1{\cdot} a)(1{\cdot} a)\)};
		\path[->]
			(epsilon) edge [loop above]             node                               {\(a,x=0,0\)} ()
			(epsilon) edge [below]                  node [sloped, near start]          {\(a,0<x<1,1\)} (s)
			(epsilon) edge [bend left=25]           node [above, sloped, near start]   {\(a,x=1,0\)} (q)
			(epsilon) edge node                               {\(a,1<x,1\)} (2)
			(s)       edge [loop left]              node [below, left]                 {\(a,0<x<1,1\)} ()
			(s)       edge [bend left=5]            node [above, sloped]               {\(a,x=1,0\)} (q)
			(s)       edge node                               {\(a,1<x,1\)} (2)
			(q)       edge [loop right]             node [right=.1cm,above,near start] {\(a,x=0,0\)} ()
			(q)       edge [bend left=15]           node [sloped]                      {\(a,0<x<1,1\)} (s)
			(q)       edge [right]                  node [right, near start]           {\(a,x=1,0\)} (r)
			(q)       edge [above]                  node [sloped]                      {\(a,1<x,1\)} (2)
			(2)       edge [loop above]             node [right, near end]             {\(a,1<x,1\)} ()
			(r)       edge [out=60, in=30, loop]    node [above]                       {\(a,x=0,0\)} ()
			(r)       edge node [sloped]                      {\(a,0<x<1,1\)} (s)
			(r)       edge [loop right]             node [below, near end]             {\(a,x=1,0\)} ()
			(r)       edge [above]                  node [sloped, near start]          {\(a,1<x,1\)} (2);
	\end{tikzpicture}
	\caption{Conjecture \(\mathcal{A}_{\mathcal{T}_{3}}\).}
	\label{fig:conjecture5}
\end{figure}
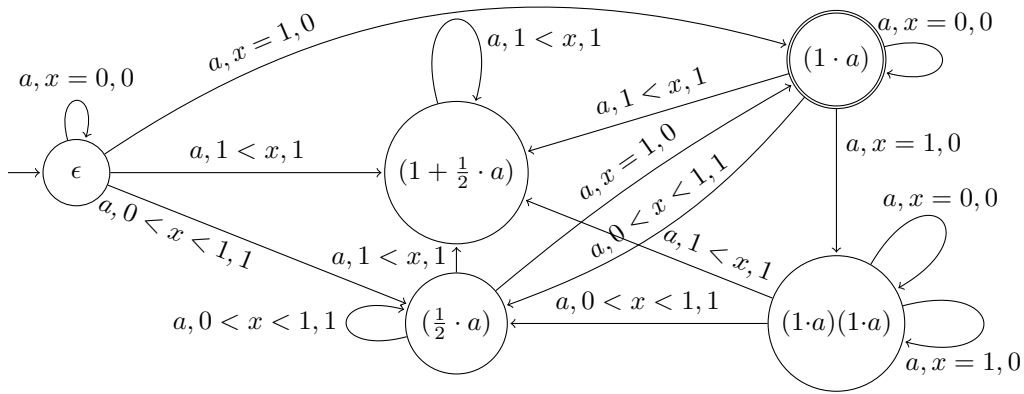

\end{document}